%% file: paper.tex
\documentclass[onecolumn, 11pt, draftclsnofoot]{IEEEtran}
\usepackage{amstext,amsmath,amssymb}
\usepackage{times}
\usepackage{graphicx}
\usepackage{algorithmic}
\usepackage{algorithm}
\usepackage{url}
\usepackage{color}
\usepackage{verbatim}

\newtheorem{theorem}{Therorem}
\newtheorem{lemma}[theorem]{Lemma}
\newtheorem{corollary}[theorem]{Corollary}
\newtheorem{definition}{Definition}

\newcommand{\bdm}{
    \begin{displaymath}}

\newcommand{\edm}{
    \end{displaymath}}

\newcommand{\be}{
    \begin{equation}}

\newcommand{\ee}{
    \end{equation}}

\newcommand{\bea}{
    \begin{eqnarray}}

\newcommand{\eea}{
    \end{eqnarray}}

\newcommand{\beit}{\begin{itemize}}
\newcommand{\eeit}{\end{itemize}}

\newcommand{\eat}[1]{ }

\ifCLASSINFOpdf
\else
\fi
\begin{document}
%
\title{Performance Guarantee under Longest-Queue-First Schedule in Wireless Networks}

\author{\IEEEauthorblockN{Bo Li, Cem Boyaci and Ye Xia\\}
\IEEEauthorblockA{Computer and Information Science and Engineering
Department, \\ University of Florida, Gainesville, Florida, USA\\
Email: \{boli,cboyaci,yx1\}@cise.ufl.edu} }
\maketitle




%
\IEEEpeerreviewmaketitle


\begin{abstract}
\input{./sections/abstract}
\end{abstract}

\section{Introduction}
\input{./sections/introduction}

\section{Preliminaries}
\input{./sections/preliminaries}

\section{Stability Region $\Omega$ under LQF}
\input{./sections/omegaregion}

\section{Stability Region $\Delta$ under LQF}
\input{./sections/deltaregion}


\section{Graph coloring and LQF scheduling}
\input{./sections/hypergraph}

\section{Experimental Examples}
\input{./sections/simulations}

\section{Conclusion}
\input{./sections/conclusion}


%

\section{Appendix}
\input{./sections/appendix}

\bibliographystyle{IEEEtran}
\bibliography{LQF}

\end{document}

%% file: sections/abstract.tex
Efficient link scheduling in a wireless network is challenging. Typical optimal algorithms require solving an
NP-hard sub-problem. To meet the challenge, one stream of research
focuses on finding simpler sub-optimal algorithms that have low
complexity but high efficiency in practice. In this paper, we
study the performance guarantee of one such scheduling algorithm,
the Longest-Queue-First (LQF) algorithm. It is known that the LQF
algorithm achieves the full capacity region, $\Lambda$, when the
interference graph satisfies the so-called local pooling
condition. For a general graph $G$, LQF
achieves (i.e., stabilizes) a part of the capacity region,
$\sigma^*(G) \Lambda$, where $\sigma^*(G)$ is the overall local
pooling factor of the interference graph $G$ and $\sigma^*(G) \leq 1$. It has been shown later that LQF achieves a larger rate region, $\Sigma^*(G)
\Lambda$, where $\Sigma^*(G)$ is a diagonal matrix. The
contribution of this paper is to describe three new achievable
rate regions, which are larger than the previously-known regions.
In particular, the new regions include all the extreme points of
the capacity region and are not convex in general. We also
discover a counter-intuitive phenomenon in which increasing the
arrival rate may sometime help to stabilize the network. This
phenomenon can be well explained using the theory developed in the
paper.


\begin{keywords}
Wireless Networks Scheduling, Longest Queue First Policy,
Stability, Local Pooling, Interference
\end{keywords}

%% file: sections/introduction.tex
\label{sec:introduction}

One of the long-standing challenges for wireless networks is how
to utilize the communication medium efficiently when links
interfere with each other. This paper is primarily concerned with
an interference model called the protocol model, where two
wireless links that interfere with each other are prohibited to
transmit data simultaneously \cite{LS06A},
\cite{balakrishnan-distance}, \cite{SMS06}. For the protocol
model, a scheduling algorithm strives to select a set of
non-interfering links for transmission in every time slot.

Finding efficient schedules can be very difficult. Tassiulas and
Ephremides \cite{Tassiulas92} showed that if the queue sizes for
the links (which are nodes in the interference graph) are viewed
as weights and a maximum weight independent set (MWIS) of the
interference graph is selected as the schedule in each time slot,
then the queues of the wireless network can be stabilized for any
arrival rate vector inside the capacity region. However, finding
an MWIS is NP-hard in general. Even under the more restricted
$k$-hop interference model, finding an MWIS is still NP-hard for
$k \geq 2$ \cite{sharma06maxweigted, SMS06}. For the 1-hop
interference model, the problem of finding an MWIS reduces to
\emph{maximum weight matching} and the complexity is $O(|V|^3)$, where $|V|$ is the number of wireless links \cite{Shroff08}. Hence, scheduling using MWIS is inapplicable to
large networks.

To reduce the complexity, some simple sub-optimal scheduling
algorithms have been introduced \cite{CLCD06, GLS07, JS07, LR06,
MEDZ06, SSBS07, Chaporkar05}. The \emph{Longest Queue First} (LQF, also known
as the greedy maximal schedule) policy is recognized for its high performance in practice \cite{Shroff08}. The LQF schedule chooses links in a decreasing order of
the queue sizes while conforming to the interference constraints.
Dimakis and Walrand showed that the LQF algorithm achieves
(stabilizes) the entire interior of the capacity region,
$\Lambda$, when the interference graph $G$ satisfies the so-called
overall local pooling condition \cite{DW06}. For general cases,
Joo {\em et al.} introduced a parameter called the overall local
pooling factor, denoted by $\sigma^*(G)$, where $0 < \sigma^*(G) \leq
1$, based on the topology of the interference graph $G$; they
showed the LQF algorithm achieves a subset of the capacity region,
$\sigma^*(G) \Lambda$ \cite{Shroff08}. Several other authors studied how to check the local pooling
condition or estimate $\sigma^*(G)$ for specific graphs
\cite{LE09, Be10, MA10}.

A single-parameter performance characterization of LQF suggests a uniform rate reduction on the links. However, it is possible that the links are subject to heterogeneous interference relations and some links can perform better than others. To capture the performance heterogeneity, a multiple-parameter
characterization of the stabilizable region by LQF was established
in \cite{BCY09}. It was shown that LQF can achieve a larger rate region,
$\Sigma^*(G) \Lambda$, where $\Sigma^*(G)$ is a diagonal matrix. Each diagonal entry of $\Sigma^*(G)$ corresponds to a link and it summarizes the link's interference constraints.



Even this multiple-parameter characterization of LQF underestimates the stability region. For instance, it excludes some parts of the capacity region that are obviously stabilizable by LQF. To progress further toward complete performance characterization, there is a need to go beyond the current framework of linear transformations on the capacity region. 
The goal of this study is to establish such a ``non-linear'' framework and expand our knowledge about the achievable rate region by LQF. The main contribution is to describe three new
achievable rate regions ($\Omega$, $\Delta_C$ and $\Delta_R$), which
are all larger than the previously-known regions. More precisely, we
show that $\Sigma^*(G) \Lambda \subseteq \Omega$ and $\Omega^o
\subseteq \Delta_C \subseteq \Delta_R$. 
Furthermore, the closures of the new regions include all the extreme points of the capacity region
and are not convex in general. This is in contrast to
previously-known regions of stability, which are
all convex and, in general, exclude some extreme points of the capacity region because they each are derived by reducing
the capacity region through a linear transformation. We
show that the new regions of stability (or their closures) are convex if and only if they are identical to the capacity region itself.
The result implies that, when LQF cannot achieve the full capacity
region, the largest achievable region, which is yet to be
discovered, cannot be convex.



The characterization of the LQF performance has been substantially improved with these new stability regions. For instance, we have found that, for an arbitrarily large $k > 0$, there are cases where an arrival rate vector $\lambda$ is outside all the previously-known stability regions but $k \lambda$ is in $\Omega$. In other words, the previously-known stability regions can underestimate the performance of LQF by an
arbitrarily large factor in certain cases, whereas the new regions can avoid such poor estimates.

The study has also yielded an interesting, counter-intuitive finding that
increasing the arrival rates may sometime help to stabilize the
network. We have discovered an example where a
rate vector achievable by LQF point-wise dominates another
rate vector not achievable by LQF. It turns out the former vector is in the stability region $\Delta_C$ whereas the latter is not.


We next summarize the key ideas of the paper.
Our theory is developed based on considering the {\em fluid limit} of an unstable network. A typical scenario is that the maximum queue size has an overall trend to grow indefinitely, which requires that, at some time $t$, a subset of the current longest queues continues to grow. From the set of the longest queues at time $t$, there is a subset that grows at the fastest rate and remains the longest in the next infinitesimal time interval. Denote this subset by $S$. Under LQF, the queues in $S$ will be served with priority in the next small time interval, which implies that the average service rate vector, when restricted to $S$, comes from the convex hull of the maximal schedules with respect to $S$. This convex hull is denoted by $Co(M_S)$. For the queues in $S$, the arrival rates must be larger than the service rates.
The discussion motivates the definition of a strictly dominating vector for a queue set $S$, which is a vector $\lambda$, when restricted to $S$, strictly dominating at least one vector in $Co(M_S)$. After removing the union of the strictly dominating vectors, where the union is over all possible subsets of the queues, we get $\Omega$.

Key to the development about $\Delta_C$ is a refinement to the notion of strictly dominating vectors, which is called uniformly dominating vectors. For the aforementioned queue set $S$, the arrival rates not only must be larger than the service rates, but also larger by the same amount, so that the queues in $S$ grow at the same rate. The removal of all the uniformly dominating vectors gives $\Delta_C$. Although the closure of $\Delta_C$ contains $\Omega$, there is value in studying and reporting the results about both regions. First, the theory about $\Omega$ provides building blocks for proving some of the results about $\Delta_C$. Second, $\Omega$ appears to be well connected to the notion of local pooling in \cite{DW06}, thus, providing some continuity in the theoretical development, whereas $\Delta_C$ does not appear so.




Throughout, we assume i.i.d. and mutually independent arrival processes. As Dimakis and
Walrand pointed out, for the same average arrival rate vector,
whether the arrival processes have zero or non-zero variances leads
to significantly different stability behavior (the former is the
case of deterministic arrivals with constant rates) \cite{DW06}.
They established a queue separation result for the case of
non-zero variances and developed a rank condition that leads to
queue separation. We generalize the rank condition. Then, we extend
$\Delta_C$ to a larger stability region $\Delta_R$ for the case of
non-zero variances. We also show the closures of $\Delta_C$ and
$\Delta_R$ are the same.

Finally, we relate the problems of finding stability conditions
under LQF to several problems in the fractional graph theory
\cite{SU97}. The latter provide tools for studying the stability
regions introduced by the paper and for characterizing the set
$\sigma$-local pooling factor given in \cite{BCY09}.

The rest of the paper is organized as follows. In Section \ref{sec:preliminaries}, we specify the models and notations. In Sections \ref{sec:omega} and \ref{sec:delta}, we introduce the $\Omega$ and $\Delta$ ($\Delta_C$ and $\Delta_R$) regions, respectively. In Section \ref{sec:hypergraph}, we introduce the fractional coloring and related problems that are relevant to the study of the stability regions. In Section \ref{sec:experiment}, we give some simulation results to confirm aspects of the theory. The conclusion is given in Section \ref{sec:conclusion}.


%% file: sections/preliminaries.tex
\label{sec:preliminaries}

In our model, a wireless network is represented by an undirected interference (or conflict) graph $G = (V, E)$, where the node set $V$ represents the set of
physical, wireless links in the network and the edge set $E$ represents the interference
relation among the physical links. Two nodes in $G$ are connected
with an edge whenever the physical links they represent interfere
with each other.\footnote{All the graphs in this paper are
interference graphs, unless specified otherwise.} We assume the node set $V$ is arbitrarily indexed from 1 to $|V|$, and hence, $V$ can be written as $V = \{1, \ldots, |V|\}$.

Given a subset
of nodes $S \subseteq V$, we denote $G_S=(S, L)$ to be the
subgraph of $G$ induced by the nodes in $S$. In other words, 
an edge $(u,v)$ belongs to $L$ if and only if $u, v \in
S$ and $(u,v) \in E$.

We assume a time-slotted system. The capacity of each wireless link is normalized to 1 per time slot. There is a queue associated with each wireless link at the transmitter. We assume single-hop traffic. Traffic arrives at the transmitter side of a link, joining the queue and waiting for transmission; after transmission, it leaves the network. We assume i.i.d. and mutually independent arrival processes to the queues. It is easy to see that, under the LQF schedule with either deterministic or typical random tie-breaking rules, the joint queue process is Markovian. By stability, we mean the Markov process is positive recurrent\footnote{Without loss of generality, we assume the Markov Chain is irreducible. See \cite{Tassiulas92} for general cases.}.

A schedule is denoted by a
$|V|$-dimensional 0-1 vector, where a value 1 in an entry
indicates the corresponding link is active and 0 otherwise. A
schedule is feasible if and only if the links that are active do
not interfere with each other. A feasible schedule is said to be
\emph{maximal} if no additional links can be activated without violating
the interference constraints. Therefore, every feasible schedule
is an independent set of $G$ and every maximal schedule is a
maximal independent set of $G$.

For the graph $G=(V, E)$, let $M_V$ denote the set of all the
maximal schedules and let $Co(M_V)$ denote the convex hull of all
the maximal schedules. When relevant, we also consider $M_V$ to be
the matrix whose columns are all the maximal schedules, with arbitrary indexing of the schedules. Similarly,
for a node-induced subgraph $G_S=(S, L)$, let $M_S$ be
the set (or matrix) representing all the maximal schedules of $G_S$
and let $Co(M_S)$ be the convex hull of all the maximal schedules in
$M_S$.

The capacity region $\Lambda$ of a network is defined as the set
of arrival rate vectors that are supportable by time sharing of
the feasible schedules. Equivalently,
\begin{align} \label{eq:capregion}
\Lambda & = \{ \lambda \ | \ 0 \leq  \lambda \leq  \mu \text{ for
some } \mu \in Co(M_V) \}.
\end{align}
For a non-empty subset of nodes $S \subseteq V$, the capacity region is
defined analogously by replacing $Co(M_V)$ with $Co(M_S)$ in
(\ref{eq:capregion}) and is denoted by $\Lambda_S$. In the above,
$\lambda \leq \mu$ means that vector $\lambda$ is component-wise
less than or equal to vector $\mu$.
The {\em interior} of the capacity region can be written as
\begin{align}
\Lambda^o = & \{ \lambda \ | 0 \leq \ \lambda < \mu \text{ for some
$\mu \in Co(M_V)$}\}.
\end{align}

The interior of the capacity region thus defined can be stabilized by the MWIS schedule and any rate vector outside the capacity region cannot be stabilized by any schedule \cite{Tassiulas92}.


Given a $|V|$-dimensional vector $\lambda$, the $|S|$-dimensional
vector $[\lambda]_S$ represents the restriction of $\lambda$ to
the set $S \subseteq V$. That is, $[\lambda]_S$ contains only
those components of $\lambda$ which correspond to the nodes in
$S$.

For a vector $\mu$ defined for a node set, let $\mu_l$ or $\mu(l)$ denote the
component associated with $l \in V$. Note that, if $\mu \in \mathbb{R}_+^{|V|}$, then the notation indicates the $l$th component of $\mu$. However, if $\mu \in \mathbb{R}_+^{|S|}$ for some non-empty $S \subseteq V$, then for $l \in S$, $\mu_l$ or $\mu(l)$ is not necessarily the $l$th component of $\mu$. If $\mu$ is any other type of vector, $\mu_i$ denotes the $i$th component of $\mu$.
We use $e$ to represent the vector
$(1, 1, \ldots, 1)'$. The dimension of the vector $e$ depends on the
context.

The capacity region for the whole graph and the capacity region
for the subset $S \subseteq V$ has the following relationship.

\begin{lemma} \label{lem:subsetcapacity}
An arrival rate vector
$\lambda \in \Lambda$ if and only if for all non-empty $S \subseteq V$,
$[\lambda]_S \in \Lambda_S$. Likewise, $\lambda \in \Lambda^o$ if
and only if for all non-empty $S \subseteq V$, $[\lambda]_S \in
\Lambda^o_S$.
\end{lemma}

\begin{proof}
Suppose $\lambda \in \Lambda$. Then, $\lambda \leq \mu$ for some
$\mu \in Co(M_V)$. It is easy to see that, for any non-empty subset $S
\subseteq V$, there must exist a vector $\nu \in Co(M_S)$ such
that $[\mu]_S \leq \nu$. Then, $[\lambda]_S \leq \nu$. Hence,
$[\lambda]_S \in \Lambda_S$. The other direction is true by taking
$S=V$. The last statement of the lemma can be proved similarly.
\end{proof}

Throughout, in the statements about rate
regions that involve topological concepts such as open/close sets and
the interior of a set, the space is assumed to be the set of
non-negative real vectors, i.e., $\mathbb{R}_+^{|V|}$. Also, in the
set-complement operation for any rate region, the whole set is
understood to be the non-negative real vectors. For a set $Y \subseteq \mathbb{R}_+^{|V|}$, we let $Y^o$ and $Y^c$ denote the interior and complement of $Y$, respectively.

In the LQF schedule, the links with longer queues are activated at
a higher priority than those with shorter queues, subject to the
interference constraints. The following may be considered as a
reference implementation of this schedule. First, one of the links
with the longest queue is selected to be in the schedule; ties are broken with either an arbitrary deterministic rule or randomly. All
links with which the selected link interferes are removed from
further consideration. Then, the same selection process repeats
over the remaining links yet to be considered until no links
remain to be considered.

\noindent  {\bf Remark: } The following is the key mathematical property about LQF that is used throughout. Suppose, at time $t$, a non-empty set $S \subseteq V$ dominates $V - S$ in the sense that, for any $i \in S$ and any $j \in V - S$, the queue size of $i$ is greater than that of $j$. Then, the schedule used at $t$ must be maximal when restricted to $S$ (i.e., with respect to $G_S$).


%% file: sections/omegaregion.tex
\label{sec:omega}

In this section, we introduce a notion of strictly dominating
vectors and construct a region denoted by $\Omega$ based on this
notion. The $\Omega$ region is larger than $\sigma^*(G) \Lambda$
and $\Sigma^*(G) \Lambda$, which have previously been shown to be
stabilizable by the LQF policy. Unlike those previously-discovered
regions of stability, the $\Omega$ region includes all the extreme
points of $\Lambda$ and it is not convex in general.


\subsection{Review of Set, Link and Overall $\sigma$-local Pooling}

Set $\sigma$-local pooling has been studied in \cite{BCY09}. It
has many interesting properties and is related to (overall)
$\sigma$-local pooling defined in \cite{Shroff08}. 


\eat{
Given a non-empty node set $S$, let
\begin{align}
\Theta_S = & \{\sigma \ | \ \sigma \mu_S \ngtr \nu_S, \text{for all
$\mu_S, \nu_S \in Co(M_S)$ } \}. \label{eq:thetaS}
\end{align}
The compliment of $\Theta_S$ is
\begin{align}
\Theta_S^c = & \{\sigma \ | \ \sigma \mu_S > \nu_S, \text{for some
$\mu_S, \nu_S \in Co(M_S)$ } \}. \label{eq:thetaSc}
\end{align}
}

\begin{definition} \label{def:setsigma}
Given a non-empty set of nodes $S \subseteq V$, the {\em set
$\sigma$-local pooling factor for $S$}, denoted by $\sigma_S^*$, is given by
\begin{align}
\sigma_S^* = & \sup \{\sigma \ | \ \sigma \mu \ngtr \nu, \text{for all
$\mu, \nu \in Co(M_S)$ } \}
\label{eq:setsigsup}\\
= & \inf \{\sigma \ | \ \sigma \mu > \nu, \text{for some
$\mu, \nu \in Co(M_S)$ } \}.
\label{eq:setsiginf}
\end{align}
\end{definition}

It has been shown that the set $\sigma$-local pooling factor is
equal to the optimal value of the following problem.
\begin{align}
\sigma_S^* 
= & \min_{\sigma, \mu, \nu} \sigma, \quad \text{subject to} \ \sigma \mu \geq \nu, \ \mu, \nu \in Co(M_S). \label{eq:setsigmapro}
\end{align}

The link $\sigma$-local pooling factor is defined as follows.
\begin{definition} \label{def:sigmaforlink}
The {\em local pooling
factor of a link} $l \in V$, denoted by $\sigma^*_l$, is given by
\begin{align}
\sigma^*_l = & \sup \{\sigma|\sigma \mu \ngtr \nu \text{ for
all } S \subseteq V \text{ such that } l \in S, 
\text{ and all } \mu, \nu \in Co(M_S)\}  \\
= & \inf \{ \sigma | \sigma \mu > \nu \text{ for some $S
\subseteq V$ such that $l \in S$, }
\text{and some } \mu, \nu \in Co(M_S) \}. \label{eq:linksiginf}
\end{align}
\end{definition}

Comparing the definitions of $\sigma^*_S$ and $\sigma^*_l$, we have
\begin{align}
\sigma^*_l=  \min_{\{S \subseteq V \ | \ l \in S\}} \sigma^*_S. \label{eq:lsigset}
\end{align}
The overall $\sigma$-local pooling factor of the graph $G = (V, E)$
is
\begin{align*}
\sigma^*(G) =  \min_{l \in V} \sigma^*_l.
\end{align*}
Let the diagonal matrix $\Sigma^*(G)$ be defined by $\Sigma^*(G)=$
diag $(\sigma^*_l)_{l\in V}$. It has been shown that
$\sigma^*(G) \Lambda$ and $\Sigma^*(G) \Lambda$ are both regions
of stability under LQF \cite{Shroff08} \cite{BCY09},
with the latter containing the former.

\subsection{Strictly Dominating Vectors and $\Omega$ Region}

We first discuss some intuition that leads to the construction of the $\Omega$ region. When the network is unstable, a typical situation is that the size of the longest queues has an overall trend of increase, if one ignores the short-time fluctuations. This would not have occurred if, for any subset $S \subseteq V$, the arrival rate is strictly less than the service rate at each node in $S$. Here, we imagine $S$ is the set of nodes with the longest queues for an extended period of time. Then, over that period of time, the schedule on each time slot must be maximal when restricted to $S$ and, by time sharing of such maximal schedules, the (average) service rate must be in $Co(M_S)$. The discussion motivates us to define the notion of strictly dominating
vectors for a subset of the nodes.



\begin{definition} \label{def:posidomin}
Given a non-empty node set $S \subseteq V$, a vector $\lambda \in \mathbb{R}_+^{|V|}$ is a strictly dominating vector of
$S$ if $[\lambda]_S > \nu$ for {\bf some} $\nu \in Co(M_S)$. The region composed with all
the strictly dominating vectors of $S$ is called the strictly
dominating region of $S$ and is denoted by $\Pi_S$. That is,
\begin{align}
\Pi_S & =\{\lambda \in \mathbb{R}_+^{|V|} \ | \ [\lambda]_S
> \nu, \text{ for some } \nu \in Co(M_S)\}. \nonumber
\end{align}
\end{definition}
For convenience, if $S = \emptyset$, we assume $\Pi_S = \emptyset$.

We are often interested in the complement of $\Pi_S$:
\begin{align}
\Pi_S^c & =\{\lambda \in \mathbb{R}_+^{|V|} \ | \ [\lambda]_S
\ngtr \nu, \text{ for all } \nu \in Co(M_S)\} \nonumber \\
& =\{\lambda \in \mathbb{R}_+^{|V|} \ | \ \text{for every } \nu \in Co(M_S), 
\text{ there exists } l \in S \text{ such that } \lambda(l) \leq \nu(l) \}.\nonumber 
\end{align}

\begin{definition}\label{def:omegaregion}
The $\Omega$ region is defined by
\begin{align}
\Omega &=\bigcap_{S \subseteq V}
 \Pi_S^c. \nonumber
 \end{align}
\end{definition}

\noindent  {\bf Remark.} A vector $\lambda$ is outside $\Omega$ if and only if
$\lambda \in \Pi_S$ for some non-empty node set $S$. Also, when restricted
to the components corresponding to the nodes in $S$, $\Pi_S$ is an
open set (it is a union of open sets). Hence, $\Omega$ is a
closed set. It can also be helpful to think $\Omega =(\bigcup_{S \subseteq V} \Pi_S)^c$.

\eat{

\begin{lemma} \label{lem:omegainter}
Consider an interference graph $G=(V, E)$. Suppose an arrival rate
vector $\lambda \in \Omega^o$. Then, for any subset $S \subseteq
V$ and any $\nu \in Co(M_S)$, there must exist $l \in S$ such that
$\lambda(l)<\nu(l)$.
\end{lemma}
\begin{proof}
Suppose there exist a subset $S \subseteq V$ and $\nu \in Co(M_S)$
such that $\lambda(l) \geq \nu(l)$ for all $l \in S$. Since
$\lambda \in \Omega^o$, $\lambda+\epsilon e \in \Omega$ for some
small enough $\epsilon>0$. Then, $[\lambda+\epsilon
e]_S>[\lambda]_S \geq \nu$. Hence, $\lambda+\epsilon e \in \Pi_S$,
which implies $\lambda+\epsilon e \not \in \Omega$ by Definition
\ref{def:omegaregion}, leading to a contradiction.
\end{proof}

}

\begin{lemma} \label{lem:omegainter}
Suppose an arrival rate
vector $\lambda$ satisfies $\lambda \in \Omega^o$. Then, for any non-empty subset $S \subseteq
V$ and any $\nu \in Co(M_S)$, there exists $l \in S$ such that
$\lambda(l) + \epsilon_o <\nu(l)$, where $\epsilon_o > 0$ is a
constant independent of $S$, $\nu$ and $l$.
\end{lemma}
\begin{proof}
Since $\lambda \in \Omega^o$, we have $\lambda+\hat{\epsilon} e \in
\Omega$ for some small enough $\hat{\epsilon} > 0$. Suppose the
conclusion of the lemma is not true. That is, suppose for any
$\epsilon > 0$, there exists a non-empty subset $S \subseteq V$ and $\nu \in
Co(M_S)$ such that $\lambda(l) + \epsilon \geq \nu(l)$ for all $l
\in S$. We can choose $\epsilon$ satisfying $0 < \epsilon < \hat{\epsilon}$. Then,
$[\lambda+\hat{\epsilon} e]_S > [\lambda + \epsilon e]_S \geq
\nu$. Hence, $\lambda+\hat{\epsilon} e \in \Pi_S$, which implies
$\lambda + \hat{\epsilon} e \not \in \Omega$ by Definition
\ref{def:omegaregion}, leading to a contradiction.
\end{proof}

\subsection{Performance Guarantee of LQF in $\Omega$ Region}

\eat{
The following assumptions will be needed later (see \cite{DW06}
for their relevance).

\noindent {\bf A1:} The arrival processes are i.i.d. and mutually independent.

\noindent {\bf A2:} (The large deviation bound on the arrival
processes) Let $A_k(n)$ be the cumulative arrivals at queue $k$ up
to time $n$, and let $\lambda_k$ be the average arrival rate at
queue $k$. For each $\epsilon>0$,
\begin{align*}
 & P(|\frac{A_k(n)}{n}-\lambda_k|>\epsilon) \leq \beta
  \exp(-n\gamma(\epsilon)) \text{ for all } n \geq 1,\\
  &\text{for some } \gamma(\epsilon)>0 \text{ and } \beta>0.
\end{align*}
 
}

\begin{theorem} \label{thm:omegastable}
If an arrival rate vector
$\lambda$ satisfies $\lambda \in \Omega^o$, then, the network is
stable under the LQF policy.
\end{theorem}

The full proof requires replicating most of the arguments in
\cite{DW06}. In the following, we only highlight the part of the
argument that needs modification.
\begin{proof}[Sketch of Proof]
Consider the fluid limit of the queue processes, denoted by
$\{q_l(t)\}_{t \geq 0}$ for all $l \in V$. For a fixed (and
regular) time instance $t$, let $S$ be the set of those longest
queues whose time derivatives at $t$, $\dot{q}_l(t)$, are the
largest under a given LQF policy instance. The queues in $S$ will remain
the longest with identical length in the next infinitesimally
small time interval. 

The service rate vector, when restricted to $S$, must
belong to the set $Co(M_S)$. Roughly, this is because $S$ contains all the queues that are the longest and remain the longest in the
near future, and hence, as remarked earlier, every LQF schedule being used must be a
maximal schedule when restricted to $S$.



Now imagine $\nu$ is the service rate vector for the nodes in $S$
at time t. Since $\lambda \in \Omega^o$, by
Lemma \ref{lem:omegainter}, there exists a link $l
\in S$ such that $\lambda(l) + \epsilon_o < \nu(l)$ for some constant $\epsilon_o > 0$.
Then, $\nu(l) - \lambda(l)> \epsilon_o$. Hence, at any time
instance, each of the longest queues decreases at a positive rate
no less than $\epsilon_o$. This is sufficient to conclude that the
original queueing process is a positive recurrent Markov process (see \cite{Dai95}),
which means the queues are stable.
\end{proof}

Next, we show $\Omega$ contains the previously-known regions of stability for LQF.
\begin{lemma} \label{lem:matrixomega}
The following holds: $\Sigma^*(G) \Lambda \subseteq \Omega$.
\end{lemma}
\begin{proof}
Consider any vector $\lambda \in \Sigma^*(G) \Lambda$. Let $S
\subseteq V$ be an arbitrary non-empty node set. Let $n \in \arg \max_{k \in
S} \sigma^*_k$. Since $\lambda \subseteq \Sigma^*(G) \Lambda$, by Lemma \ref{lem:subsetcapacity} and (\ref{eq:lsigset}),
$[\lambda]_S \in \sigma^*_n \Lambda_S \subseteq \sigma^*_S
\Lambda_S$. Then, there exists a vector $\mu \in \Lambda_S$ such
that $[\lambda]_S \leq \sigma^*_S \mu$. According to Definition
\ref{def:setsigma}, $\sigma^*_S \mu \not > \nu$ for any $\nu \in
Co(M_S)$. Hence, $[\lambda]_S \not
> \nu$ for any $\nu \in Co(M_S)$, implying $\lambda \in \Pi_S^c$. Since
$S$ is chosen arbitrarily, $\lambda \in \bigcap_{S \subseteq V}
\Pi_S^c=\Omega$ by Definition \ref{def:omegaregion}.
\end{proof}

\subsection{Shape of $\Omega$ Region}

The previously-known regions of stability under LQF, such as $\Sigma^*(G) \Lambda $, are derived by
reducing the capacity region through a linear transformation.
Since the capacity region $\Lambda$ is convex, each of these derived stability regions
is also convex. In contrast, we will show that the shape of the $\Omega$ region is
not convex in general. Furthermore, when the previously-known regions are not identical to $\Lambda$, they exclude many, if not most, of the extreme points of $\Lambda$. We will show that $\Omega$ contains all the extreme points of $\Lambda$.


\eat{

\begin{lemma}
Given a graph $G=(V, E)$, let $\sum^*(G)=$ diag $(\sigma^*_l)_{l\in
V}$. The region $\sum^*(G)\Lambda$ is convex.
\end{lemma}
\begin{proof}
Suppose $\mu, \nu \in \sum^*(G)\Lambda$. Then, there must exist an
$\alpha, \beta \in \Lambda$ such that $\mu=\sum^*(G) \alpha$ and
$\nu=\sum^*(G) \beta$. For any $0 \leq f \leq 1$, $f \mu+ (1-f) \nu=
f \sum^*(G) \alpha +(1-f) \sum^*(G) \beta= \sum^*(G) [f
\alpha+(1-f)\beta]$. Since $\Lambda$ is convex, $f \alpha+(1-f)\beta
\in \Lambda$ which implies that $f \mu+ (1-f) \nu \in
\sum^*(G)\Lambda$.
\end{proof}

}


\begin{lemma} \label{indepextrem}
The set of all the independent
sets of the interference graph $G$, i.e., the set of all the feasible schedules, has a bijection to the set of all the extreme points of $\Lambda$.
\end{lemma}
Note that we consider the empty schedule where no link is activated a trivial independent set. Lemma
\ref{indepextrem} establishes a connection between the graph
topology and the geometry of $\Lambda$ in a vector space. The proof
can be found in the Appendix.

\eat{
(see Fig. \ref{fig:extrindep}). The proof can be found in the
Appendix.

\begin{figure}[htbp]
  \begin{center}
    \includegraphics[width=2.8in]{fig/extremindep.eps}
  \end{center}
  \caption{Relationship between graph independent sets and extreme points of the capacity region. Left:
  the graph $G$; Right: the capacity region.}
  \label{fig:extrindep}
\end{figure}
}
\begin{lemma} \label{lem:phiids}
Suppose $\lambda$ is a vector
corresponding to an independent set of the interference graph $G$. Then, $\lambda \in
\Omega$.
\end{lemma}
\begin{proof}
$[\lambda]_S$ is an independent set of the node-induced subgraph
$G_S$ for any non-empty $S \subseteq V$. Then, $[\lambda]_S \not >
\nu$ for any $\nu \in Co(M_S)$, which implies $\lambda \not \in
\Pi_S$. Since $S$ is arbitrary, we must have $\lambda \in \Omega$.
\end{proof}

\begin{corollary} \label{cor:extpt}
All the extreme points of the capacity
region $\Lambda$ belong to $\Omega$.
\end{corollary}
\begin{proof}
This is a result of Lemma \ref{indepextrem} and
Lemma \ref{lem:phiids}.
\end{proof}


%

\begin{figure}[htbp]
  \begin{center}
    \includegraphics[width=1.0in]{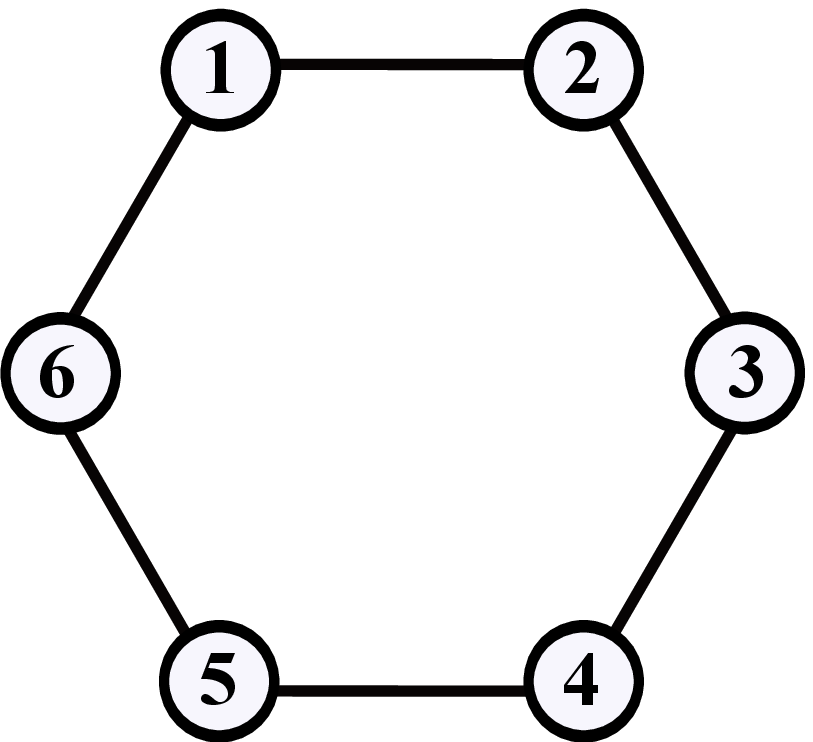}
  \end{center}
  \caption{The six-cycle graph, $C_6$.}
  \label{fig:sixcycle}
\end{figure}

As an example, let $G$ be the six cycle graph in Fig.
\ref{fig:sixcycle}. The arrival rate vector
$\lambda=(1,0,1,0,1,0)'$ corresponds to an independent set, and
hence, $\lambda \in \Omega$. However, we know that
$\Sigma^*(G)=\text{diag}(2/3,$ $2/3, 2/3, 2/3, 2/3, 2/3)$. As a
result, $\lambda \not \in \Sigma^*(G) \Lambda$. The example shows
that $\Omega$ can be strictly larger than $\Sigma^*(G) \Lambda$.
In the example, $\Omega - \Sigma^*(G)\Lambda$ contains
not only the extreme points. For instance, one can check that, for
$\lambda=(7/10, 1/10, 7/10, 1/10, 7/10, 1/10)'$, $\lambda \in
\Omega$ but $\lambda \not \in \Sigma^*(G) \Lambda$.


The next example shows that, the previously-discovered stability regions $\sigma^*(G) \Lambda$ and $\Sigma^*(G)
\Lambda$ can underestimate the performance of LQF by an
arbitrarily large factor in certain directions and in certain cases, whereas $\Omega$ can avoid such poor estimates.

\begin{lemma} \label{lem:arbitraryimprove}
For any $k > 0$, there exists an interference graph
$G=(V, E)$ and an arrival rate vector $\lambda$ such that $\lambda
\notin \Sigma^*(G) \Lambda$, but $k \lambda \in \Omega$.
\end{lemma}

\begin{proof}
Consider the bipartite graph in Fig.
\ref{fig:almost.complete.bipartite} with $N$ pairs of nodes, where $N=4$ in this particular case. It is almost a complete
bipartite graph except that every corresponding pair of nodes (such as nodes 1
and 2) does not have an edge between them. It is easy to check that
$\Sigma^*(G)=\text{diag}(2/N, 2/N,..., 2/N)$.
Therefore, the rate vector $\lambda=(2/N+\epsilon, 0, 0,..., 0)'$, where $\epsilon > 0$, is not
in $\Sigma^*(G) \Lambda$.
For any $k>0$, we can find a large enough $N$ and a small enough $\epsilon$ such that $k(2/N+\epsilon)\leq 1$. Then, we have $k
\lambda=(k(2/N+\epsilon), 0, 0,...,0)'$ in $\Omega$.
\end{proof}

\begin{figure}[htbp]
  \begin{center}
    \includegraphics[width=0.8in]{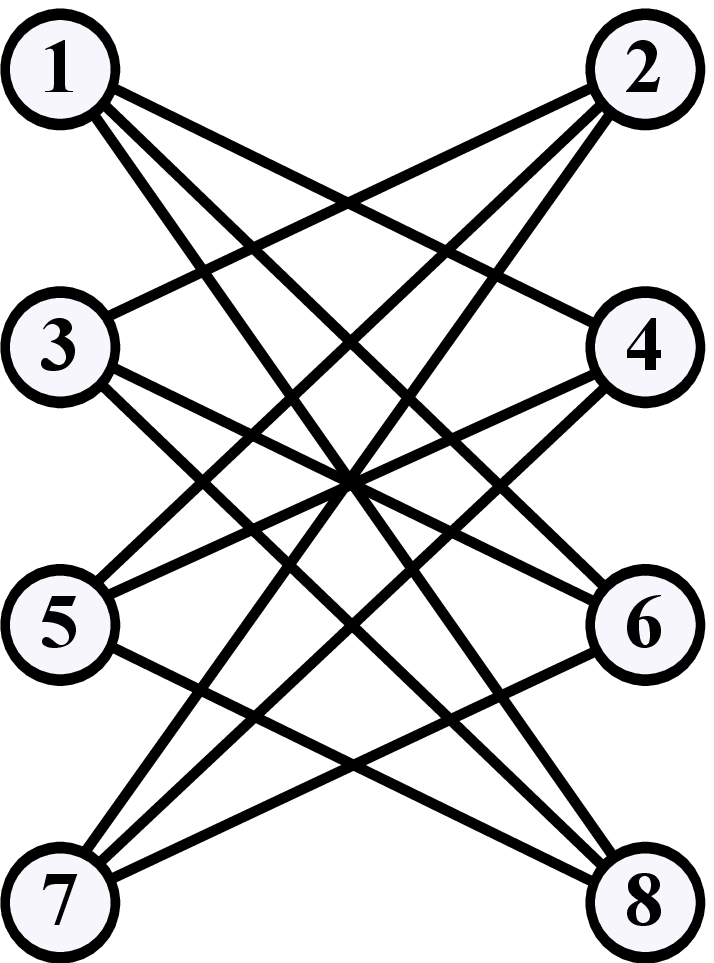}
  \end{center}
  \caption{A bipartite graph.}
  \label{fig:almost.complete.bipartite}
\end{figure}

\begin{figure}[htbp]
  \begin{center}
    \includegraphics[width=2in]{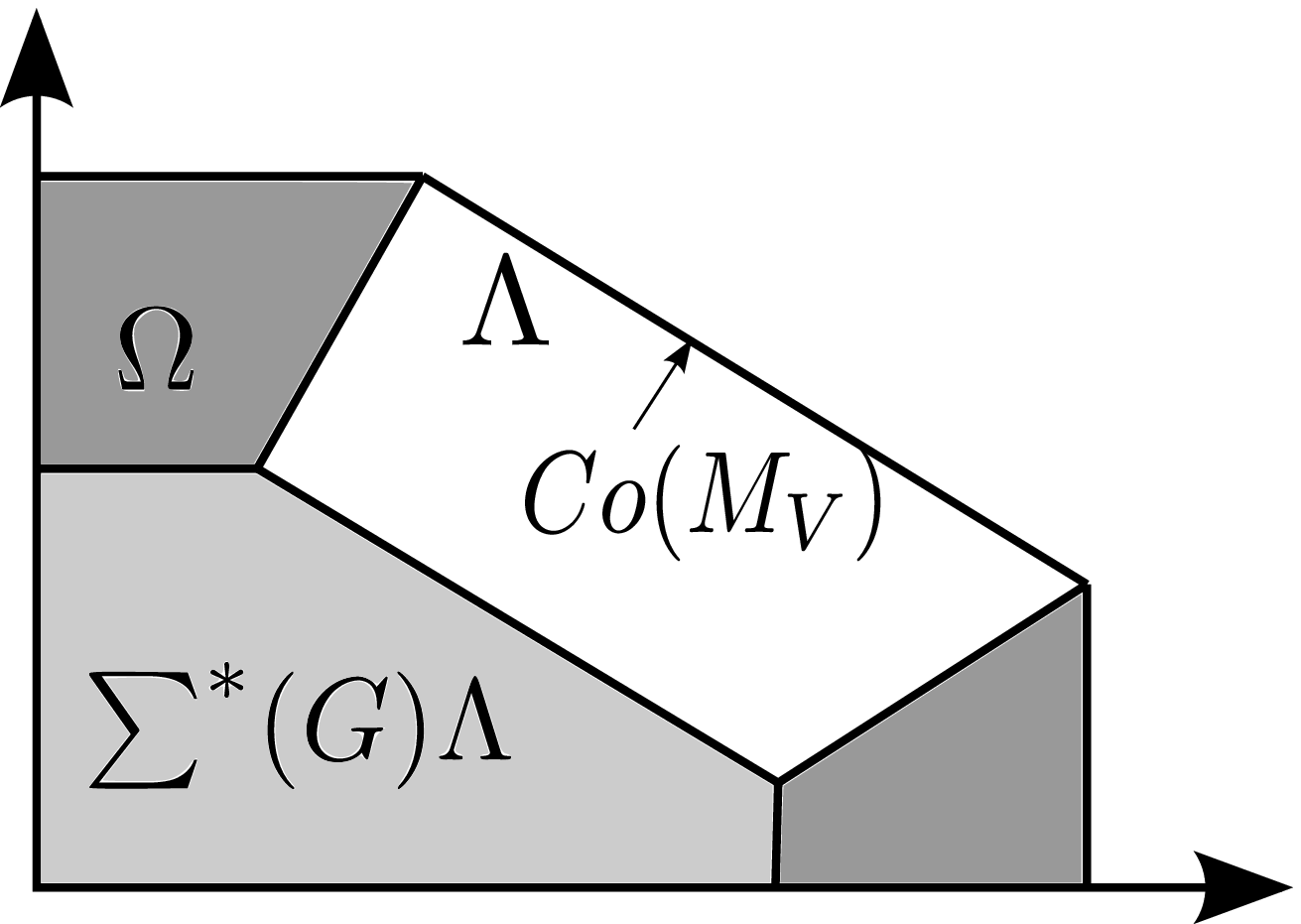}
  \end{center}
  \caption{The $\Omega$ region and other relevant regions. The largest convex polytope is $\Lambda$. The entire shaded region is $\Omega$, which is not a convex set.}
  \label{fig:omega}
\end{figure}

Though we cannot draw various regions in a high-dimensional vector
space, it may still be helpful to make a highly simplified
illustration with Fig. \ref{fig:omega}. The whole capacity region
$\Lambda$ is convex. The region $\Sigma^*(G) \Lambda$ is derived by
scaling down the capacity region $\Lambda$ using the diagonal matrix
$\Sigma^*(G)$. This sort of scaling
usually cuts off many or most extreme points of $\Lambda$. The newly
defined stability region $\Omega$ is a superset of $\Sigma^*(G)
\Lambda$ and $\Omega$ contains all the extreme points of $\Lambda$.
The figure makes the point that $\Omega$ is not convex in general.
We next show $\Omega$ is convex if and only if it is equal to
$\Lambda$.

\begin{lemma} \label{lem:convexcond}
The following statements are equivalent.

1. $\Omega$ is a convex.

2. $G$ is an overall local pooling graph.

3. $\Omega=\Lambda$.
\end{lemma}
\begin{proof}
First, we prove that statement 1 implies statement 2. Suppose $G$
is not overall local pooling. We claim that there must exist a non-empty
set $S \subseteq V$ and $\mu, \nu \in Co(M_S)$ such that $\mu>\nu$. Since $G$
is not overall local pooling, there exists a non-empty set $S
\subseteq V$ such that $\sigma^*_S < 1$, which implies that there
exist $\mu, \nu \in Co(M_S)$ and $\sigma^*_S \mu \geq \nu$, according to (\ref{eq:setsigmapro}). If $\nu
> 0$, we have the required set $S$, and $\mu, \nu \in Co(M_S)$ with
$\mu > \nu$. If not, let $H = \{l \in S | \nu(l) > 0\}$.
Because $\nu \in
Co(M_S)$ and $[\nu]_{S-H}=0$, it is easy to show $[\nu]_H \in Co(M_H)$\footnote{Suppose we write $M_S = (m^i)_{i=1}^{|M_S|}$, where each $m^i$ is a maximal schedule with respect to $S$. We can represent $\nu$ as $\nu = \sum_{i = 1}^{|M_S|} \alpha_i m^i$, where $\sum_i \alpha_i=1$ and $\alpha_i \geq 0$ for all $i$. Since $[\nu]_{S-H}=0$, we have $[m^i]_{S-H} = 0$ for each $i$. It is clear that $[m^i]_H$ corresponds to an independent set of $G_H$, the subgraph of $G$ induced by $H$. Moreover, by the maximality of $m^i$ with respect to $S$, if $m^i(j) = 0$ for some $j \in H$, it must be that $m^i(k) = 1$ for some $k \in H$ and $j$ and $k$ interfere with each other, i.e., $(j, k) \in E$. Therefore, $[m^i]_H$ must be maximal with respect to $H$. Hence, by $[\nu]_H = \sum_{i = 1}^{|M_S|} \alpha_i [m^i]_H$, we get $[\nu]_H \in Co(M_H)$.}. Because $\mu \in Co(M_S)$ and
$H \subseteq S$, there must exist $\tilde{\mu} \in
Co(M_H)$ such that $\tilde{\mu} \geq [\mu]_H$.
Then, $\sigma^*_S \tilde{\mu} \geq \sigma^*_S [\mu]_H \geq
[\nu]_H > 0$. Thus, $\tilde{\mu} > [\nu]_H$ and
$\tilde{\mu}, [\nu]_H \in Co(M_H)$. By renaming
$H$ to be $S$, $\tilde{\mu}$ to be $\mu$ and
$[\nu]_H$ to be $\nu$, we have the required set $S$ and
$\mu, \nu \in Co(M_S)$ with $\mu>\nu$.

Let $\lambda \in \mathbb{R}_+^{|V|}$ be an extended vector from $\mu$ such that
$[\lambda]_S=\mu$ and $[\lambda]_{V-S}=0$. According to Definition
\ref{def:posidomin} and \ref{def:omegaregion}, $\lambda \not \in
\Omega$. Since $\mu \in Co(M_S)$, we can write $\mu = \sum_{i=1}^{|M_S|}
\alpha_i m^i$, where $\sum_i \alpha_i=1$ and $\alpha_i \geq 0$ for
all $i$, and $m^i$ for $i = 1, \ldots, |M_S|$ are all the maximal schedules with respect to $S$. For each $i$, let $\tilde{m}^i$ be a $|V|$-dimensional
vector extended from $m^i$, such that $\tilde{m}^i(j)=m^i(j)$ when
$j \in S$ and $\tilde{m}^i(j)=0$ when $j \not \in S$. Clearly,
each $\tilde{m}^i$ corresponds to an independent set of $G$. Hence, by Lemma \ref{lem:phiids}, $\tilde{m}^i \in \Omega$ for all $i$.
Since $\lambda=\sum_{i=1}^{|M_S|} \alpha_i \tilde{m}^i$ and
$\lambda \not \in \Omega$, we conclude that $\Omega$ is not
convex.

Next, we show that statement 2 implies statement 3. Since $G$ is
an overall local pooling graph, $\Sigma^*(G)= I$ (the identity
matrix). By Lemma \ref{lem:matrixomega}, $\Lambda \subseteq
\Omega$. Hence, $\Omega= \Lambda$.

Finally, statement 3 implies statement 1 since $\Lambda$ is
convex.
\end{proof}

\noindent  {\bf Remark.} Suppose, for a given interference graph, the LQF algorithm does not achieve
the full interior of the capacity region. Lemma \ref{lem:convexcond} implies that $\Omega$ is not convex. Furthermore, since the closure of the full
stability region of LQF (which is unknown) contains $\Omega$, it contains all the extreme points of the capacity region $\Lambda$. Hence, the closure of the full stability region of LQF cannot be convex either, and it cannot be characterized by any linear transformation of the capacity region.

%% file: sections/deltaregion.tex
\label{sec:delta}

In this section, we develop a notion termed as uniformly
dominating vectors. It leads to a stability region
$\Delta_C$, which is a superset of $\Omega^o$. When the arrival
processes are not constant, i.e., when the variances of the i.i.d.
arrival processes are non-zero, we obtain a stability region
$\Delta_R$, which contains $\Delta_C$.

\subsection{Motivating Examples} \label{sec:examples}


\noindent {\bf Example 1:} We will first give an example to show
that an arrival rate vector $\lambda \not \in \Omega$ can sometime
be stabilized by LQF. Hence, there is a region larger than $\Omega$ that
captures the performance of LQF more precisely. The example also contains hints about how such a region can be defined.

Consider the six cycle graph $G$ in Fig. \ref{fig:sixcycle}. There are exactly five maximal schedules: $s^1=(1, 0, 1,
0, 1, 0)'$, $s^2=(0, 1, 0, 1, 0, 1)'$, $s^3=(1, 0, 0, 1, 0, 0)'$,
$s^4=(0, 1, 0, 0, 1, 0)'$, $s^5=(0, 0, 1, 0, 0, 1)'$. Suppose the
arrival rate vector is $\lambda=(5/12+\epsilon, 1/3+\epsilon,
1/3+\epsilon, 1/3+\epsilon, 1/3+\epsilon, 1/3+\epsilon)'$, where
$\epsilon>0$ is some small enough constant. Let $e=(1, 1, 1, 1, 1,
1)'$, $\mu=\frac{1}{2} e$ and $\nu=\frac{1}{3} e$. Then, one can
check that $\mu=\frac{1}{2} s^1+ \frac{1}{2} s^2$ and
$\nu=\frac{1}{3} s^3+ \frac{1}{3}s^4+ \frac{1}{3}s^5$, which
implies that $\mu, \nu \in Co(M_V)$. For $0<\epsilon<1/12$, $\nu <
\lambda < \mu$. Hence, $\lambda \in \Lambda^o$ and $\lambda \not
\in \Omega$ by Definition \ref{def:posidomin} and
\ref{def:omegaregion}.

Consider the fluid limit of the queue processes under LQF,
denoted by $\{q_l(t)\}_{t \geq 0}$, for all $l \in V$. For a fixed
(regular) time instance $t$, let $S$ be the set of those longest
queues whose time derivatives at $t$, $\dot{q}_l(t)$, are the
largest. The queues in $S$ will remain the longest with identical
length in the next infinitesimally small time interval. Since
$\lambda \in \Lambda^o$, $[\lambda]_S \in \Lambda^o_S$ by Lemma
\ref{lem:subsetcapacity}. If $S \neq V$, it is a fact that
the node-induced subgraph $G_S$ satisfies the local pooling
condition \cite{BCY09}. An argument similar to that in the proof of Theorem
\ref{thm:omegastable} shows that the queues in $S$ all have a
negative drift.


The case of $S=V$ is more subtle. Since only the maximal
schedules of $G$ are used during the
aforementioned infinitesimally small time interval, we can
assume that the service rate vector is $\gamma=\sum_{i=1}^{5}
\alpha_i s^i$, where $\sum_{i=1}^{5} \alpha_i=1$ and $\alpha_i \geq
0$ for all $i$. In the fluid
limit, $\dot{q}_l(t)=\lambda_l-\gamma_l$ for $l \in V$. By assumption, $\dot{q}_l(t)$ should be identical for all nodes
$l \in V$. However, one can check that it is impossible to find such $\gamma$ for the given $\lambda$. Therefore, the case of
$S=V$ would not have occurred, and only the case of $S \neq V$ needs to be considered. Hence, $G$ is stable under LQF for the given $\lambda$,
according to the discussion for the $S \neq V$ case.

\noindent {\bf Example 2:} Let $\lambda^1 =$ $0.7 (1/2-\epsilon$,
$1/2-\epsilon$, $1/2-\epsilon$, $1/2-\epsilon$, $1/2-\epsilon$,
$1/2-\epsilon)'$ and $\lambda^2=(1/2-\epsilon, 1/2-2\epsilon,
1/2-2\epsilon, 1/2-2\epsilon, 1/2-2\epsilon, 1/2-2\epsilon)'$ and
$\epsilon=10^{-3}$. Both $\lambda^1$ and $\lambda^2$ are outside
$\Omega$. 
Interestingly, although $\lambda^1 < \lambda^2$, $\lambda^1$ cannot be stabilized by LQF while $\lambda^2$ can. This has been verified by simulation experiments under constant arrivals. We
will next develop a theory that provides a larger stability region
and also can explain this counter-intuitive example.

\subsection{Uniformly Dominating Vector and $\Delta_C$ Region}

\begin{definition} \label{def:uniformposidomin}
Given a non-empty node set $S \subseteq V$,
a vector $\lambda \in \mathbb{R}_+^{|V|}$ is said to be a uniformly dominating
vector of $S$ if $[\lambda]_S=\nu +d e$ for some $\nu \in
Co(M_S)$ and scaler $d \geq 0$. The region composed with all the
uniformly dominating vectors of $S$ is called the
uniformly dominating region of $S$ and is denoted by
$\Gamma_S$. That is, 
\begin{align}
\Gamma_S &= \{\lambda \in \mathbb{R}_+^{|V|} \ | \ [\lambda]_S =\nu+ d e, \text{ for some } \nu \in Co(M_S) 
\text{ and some scalar } d \geq 0\}. \nonumber
\end{align}
\end{definition}
By convention, if $S = \emptyset$, we assume $\Gamma_S = \emptyset$.

\begin{definition} \label{def:deltacregion}
The $\Delta_C$ region is defined by
\begin{align}
\Delta_C &= \bigcap_{S \subseteq V} \Gamma_S^c. \nonumber
\end{align}
\end{definition}
\noindent  {\bf Remark.} Note that a vector $\lambda$ is outside $\Delta_C$ if and only if
$\lambda \in \Gamma_S$ for some non-empty node set $S$.

\begin{lemma} \label{gammacloseddeltaopen}
For any non-empty $S \subseteq V$, $\Gamma_S$ is closed. Hence,
$\Delta_C$ is open.
\end{lemma}
\begin{proof}
Let $B=\{d e| d \geq 0\}$, where $e$ is $|S|$ dimensional, and let 
$C=Co(M_S)$. It is easy to see $C$ is compact and $B$ is closed.
From Definition \ref{def:uniformposidomin}, $\Gamma_S$ is $B+C$
extended to the $|V|$-dimensional space. It can be shown that
$B+C$ is closed, and hence, $\Gamma_S$ is closed. Then, $
\Delta_C=\bigcap_{S \subseteq V} \Gamma_S^c$ is open (with respect to the metric space $\mathbb{R}_+^{|V|}$).
\end{proof}

\begin{lemma} \label{lem:deltainter}
Suppose $\lambda \in \Delta_C$ and suppose $S \subseteq V$ is a non-empty node set. If
$\nu-[\lambda]_S=d e$ for some $\nu \in Co(M_S)$, then $d
>\epsilon_o$, for some $\epsilon_o>0$ independent of $S$ and $\nu$.
\end{lemma}
\begin{proof}
Suppose $\nu-[\lambda]_S=de$ (here, $e$ is of $|S|$-dimension) for
some $\nu \in Co(M_S)$. Since $\lambda \in \Delta_C$ and
$\Delta_C$ is open, $\lambda+\epsilon_o e \in \Delta_C$ (here, $e$
is of $|V|$-dimension) for some small enough $\epsilon_o>0$
independent of $S$ and $\nu$. Then $\nu-[\lambda+\epsilon_o
e]_S=(d-\epsilon_o)e$ or $[\lambda+\epsilon_o
e]_S=\nu+(\epsilon_o-d)e$. Since $\lambda+\epsilon_o e \in
\Delta_C$, $\lambda+\epsilon_o e \not \in \Gamma_S$. Hence,
$\epsilon_o-d<0$ or $d>\epsilon^o$.
\end{proof}


%
%

The constant $\epsilon_o$ will serve as a bound for the rate of the Lyapunov drift in the performance analysis.

\subsection{Performance Guarantee of LQF in $\Delta_C$ Region}

\begin{theorem} \label{thm:deltacstability}
If an arrival rate vector
$\lambda$ satisfies $\lambda \in \Delta_C$, then, the network is
stable under the LQF policy.
\end{theorem}
\begin{proof}[Sketch of Proof]
%
Again, consider the fluid limit of the queue process and apply a
similar argument as in the proof of Theorem \ref{thm:omegastable}.
Let $S \subseteq V$ be the set of nodes whose queues are the longest at time
$t$ and will remain the longest for the next infinitesimally small
time interval. Let $\nu_S$ be the service rate vector for the
nodes in $S$ at time $t$. Under LQF, $\nu_S \in Co(M_S)$ and
$\nu_S-[\lambda]_S=\epsilon e$ for some $\epsilon$. Since $\lambda
\in \Delta_C$, by Lemma \ref{lem:deltainter}, we have $\epsilon>\epsilon_o$ for some $\epsilon_o>0$ independent of $S$
and $\nu$. Hence, at any time instance, each of the longest queues
decreases at a positive rate no less than $\epsilon_o$. This is
sufficient to conclude that the original queueing process is a
positive recurrent Markov process, which means the queues are
stable.
\end{proof}

\begin{lemma}
$\Omega^o \subseteq \Delta_C$.
\end{lemma}
\begin{proof}
Suppose $\Omega^o \not \subseteq \Delta_C$. Then, there exists a
vector $\lambda \in \Omega^o$ and $\lambda \not \in \Delta_C$.
Hence, $[\lambda]_S=\nu+d e $ for some non-empty $S \subseteq V$,
$\nu \in Co(M_S)$ and $d \geq 0$. Since $\lambda \in \Omega^o$,
$\lambda+\epsilon e \in \Omega$ for some small enough
$\epsilon>0$. From $[\lambda+\epsilon e]_S=\nu+(d+\epsilon)
e$ and $d+\epsilon >d \geq 0$, we have $[\lambda+\epsilon
e]_S>\nu$. Hence, $\lambda+\epsilon e \not \in \Omega$, leading to
a contradiction.
\end{proof}


Consider Example 2 in Section \ref{sec:examples}. With the linear
programming tools introduced in Section \ref{sec:hypergraph}, one can check
that $\lambda^2 \in \Delta_C$ but $\lambda^1 \not \in
\Delta_C$. This explains why $\lambda^2$ can be stabilized by LQF
while $\lambda^1$ cannot, even though $\lambda^2 > \lambda^1$.


\subsection{Rank Condition and $\Delta_R$ Region }

For the same average arrival rate vector, whether the i.i.d. arrival
processes have zero or non-zero variances leads to significantly
different stability behavior (in the former case, the arrival
processes are deterministic with constant rates). This issue has been
discussed in \cite{DW06} where the authors develop a queue
separation result related to a rank condition about the matrices
of the maximal independent sets. We next generalize the rank
condition. Then, we extend $\Delta_C$ to a larger stability region
$\Delta_R$. We will show $\Delta_R$ can be stabilized under LQF
when the arrival processes all have non-zero variances.


\begin{definition} \label{def:rankingcond}
Let $S \subseteq V$ be a non-empty set. We call the matrix $(M_S, e)$ the extended schedule matrix for $S$ (or graph $G_S$). Let $R(M_S, e)$ denote the rank of the extended
schedule matrix, i.e., the number of linearly independent columns
in the matrix $(M_S, e)$. We say $S$ (or graph $G_S$) has a high rank if $R(M_S, e) = |S|$. Otherwise, we
say $S$ (or $G_S$) has a low rank.
\end{definition}

Suppose $S \subseteq V$ is the set of nodes with the longest queues
at some time instance. When the arrival has non-zero variances, the queue separation result suggests (Lemma 1 and Lemma 3 of
\cite{DW06}): If the rank $R(M_S) \leq |S|-2$, then, with
probability 1, the queue sizes of $S$ will not stay identical in the
next infinitesimal time interval. We find that the condition $R(M_S)
\leq |S|-2$ can be relaxed to $R(M_S, e) \leq |S|-1$, i.e., the low
rank condition in Definition $\ref{def:rankingcond}$. The queue
separation lemma (Lemma 1 of
\cite{DW06}) uses the assumption $R(M_S) \leq |S|-2$ to obtain a
vector $\nu$ such that $\nu' e=0$ and $\nu' M_S=0$. Such a vector
$\nu$ still exists when the low rank condition in Definition $\ref{def:rankingcond}$ is satisfied. Then,
every subsequent step in the proof of the queue separation lemma
still holds. The low rank condition is a generalization since
$R(M_S) \leq |S|-2$ implies $R(M_S, e) \leq |S|-1$.

Roughly speaking, when the variances are non-zero, the randomness in the arrival processes pressures the queues in $S$ to move around in an $|S|$-dimensional space. This means that the $|S|$ queues cannot be simultaneously the longest queues for a sustained period of time (in which case, the queue trajectory moves along a line), unless the service can fully compensate the pressure from the arrival processes. But, full compensation is not possible in the low-rank case since the service rate vector lives in a lower-dimensional space. What will happen is that some subset of the queues in $S$ with a high rank will dominate the rest. This is known as queue separation. The implication is that, in the case of non-zero variances, there is no need to consider the low-rank subsets of $V$ when evaluating the performance degradation of LQF. The discussion motivates the following definition of $\Delta_R$.

\begin{definition} \label{def:deltasregion}
The $\Delta_R$ region is defined by
\begin{align}
\Delta_R &= \bigcap_{S \subseteq V, S \text{ with high rank}}
\Gamma_S^c. \nonumber
\end{align}
\end{definition}
In words, a vector $\lambda$ is outside $\Delta_R$ if and only
if $\lambda \in \Gamma_S$ for some node set $S$ that has a high
rank.

By comparing the definitions of
$\Delta_C$ and $\Delta_R$, we have the following lemma.
\begin{lemma} \label{lem:deltacs}
$\Delta_C \subseteq \Delta_R$.
\end{lemma}

In addition to the i.i.d and mutually independent assumptions, the following assumption on the arrival processes is needed for technical reasons (see \cite{DW06} for their relevance).


\noindent {\bf A1:} (The large deviation bound on the arrival
processes) Let $A_l(n)$ be the cumulative arrivals at queue $l$ (at node $l \in V$) up to time $n$, and let $\lambda_l$ be the average arrival rate at
queue $l$. For each $\epsilon>0$,
\begin{align*}
 & P(|\frac{A_l(n)}{n}-\lambda_l|>\epsilon) \leq \beta
  \exp(-n\gamma(\epsilon)) \text{ for all } n \geq 1,
  \text{ for some } \gamma(\epsilon)>0 \text{ and } \beta>0.
\end{align*}

 
\begin{theorem} \label{thm:deltasstability}
Assume the condition in A1 holds and assume the variance of
the i.i.d arrival process to each node is non-zero but finite. If
an arrival rate vector $\lambda$ satisfies $\lambda \in \Delta_R$,
then, the network is stable under the LQF policy.
\end{theorem}
\begin{proof}[Sketch of Proof]
Again, consider the fluid limit of the queue process and apply a
similar argument as in the proof of Theorem \ref{thm:omegastable}.
Let $S \subseteq V$ be the set of nodes whose queues are the longest at time
$t$ and will remain the longest for the next infinitesimally small
time interval. By replicating most of the arguments in the queue separation lemmas (Lemma 1 and Lemma 3 in \cite{DW06}), it can be shown that $S$
must have a high rank\footnote{The only change is to Lemma 3 in \cite{DW06}. Instead of saying for any low-rank set, there must be a subset that satisfies local pooling, we say for any low-rank set, there must be a subset that is of high rank. This is so because a set with a single node is of high rank. The modification is needed in the proof of Lemma 3 in \cite{DW06}. The statement of Lemma 3 also needs to be modified accordingly.}. Otherwise, the queue sizes of the nodes in
$S$ will be separated and they cannot all remain the longest.
Hence, we can apply the same argument as that in Theorem
\ref{thm:deltacstability}, but only to the high-rank node sets.
\end{proof}

Some graph examples are given in Fig. \ref{fig:lqfregion}, regarding
their set $\sigma$-local pooling factors and ranks. Note that, the
shaded region includes those subsets $S$ which either satisfy
$\sigma^*_S=1$, i.e., set local pooling (SLoP), or have low rank. Those subsets need not to be considered for the performance of LQF in case of non-zero variances.

\eat{ {\bf Ye: In Fig. \ref{fig:lqfregion}, what does the label
``Local Pooling'' mean? If I understand it correctly, the graphs
inside the oval have set local pooling factor equal to 1; outside,
the factor is less than 1. Is this correct? If so, please change the
label appropriately, e.g., set local pooling and non-set local
pooling. Also, are there any more comments you wish to say about
this figure besides the above short description?} }

\begin{figure}[htbp]
  \begin{center}
    \includegraphics[width=2.9in]{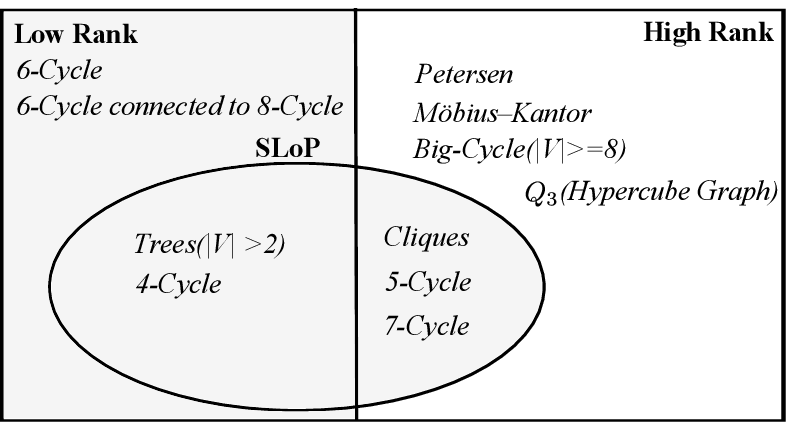}
  \end{center}
  \caption{Graph examples and classification by the set $\sigma$-local pooling factor and rank condition.
  For graphs $G_S=(S, L)$ inside the oval, $\sigma_S^* = 1$; outside the oval, $\sigma_S^* < 1$.
  The graph labeled `6-cycle connected to 8-cycle' is shown in Fig. \ref{fig:eight.and.six.cycles}.}
  \label{fig:lqfregion}
\end{figure}

\subsection{Further Properties of Regions $\Delta_C$ and $\Delta_R$}

It has been demonstrated that $\Delta_C \subseteq
\Delta_R$. We now continue to study the properties of the two
regions and their relationship.

\eat{

\begin{lemma} \label{lem:omegaclosed}
Region $\Omega$ is a closed.
\end{lemma}
\begin{proof}
By definition \ref{def:omegaregion}, $\Omega=\bigcap_{S \subseteq
V} \Pi_S^c=(\bigcup_{S \subseteq V} \Pi_S)^c$. Since
$\Pi_S=\{\lambda| \lambda \geq 0, [\lambda]_S > \nu, \text{for
some } \nu \in Co(M_S)\}$, $\Pi_S$ is open.  Then, $\bigcup_{S
\subseteq V} \Pi_S$ is open which implies that $\Omega=(\bigcup_{S
\subseteq V} \Pi_S)^c$ is closed.
\end{proof}

}

\begin{theorem} \label{thm:deltacr}
The closures of $\Delta_C$ and $\Delta_R$ are the same, i.e.
$\overline{\Delta_C}=\overline{\Delta_R}$.
\end{theorem}

\begin{proof}
Since $\Delta_C \subseteq \Delta_R$, we have $\overline{\Delta_C}
\subseteq \overline{\Delta_R}$. We will next show
$\overline{\Delta_R} \subseteq \overline{\Delta_C}$. Since
$\Delta_R=\Delta_C \bigcup (\Delta_R-\Delta_C)$ and
$\overline{\Delta_R}=\overline{\Delta_C} \bigcup
\overline{(\Delta_R-\Delta_C)}$, we only need to show
$\overline{(\Delta_R-\Delta_C)} \subseteq \overline{\Delta_C}$.

Given any vector $\tilde{\lambda} \in \Delta_R-\Delta_C$, by
comparing Definition \ref{def:deltacregion} and
\ref{def:deltasregion}, we have 
\begin{align}
\tilde{\lambda} \in \bigcap_{\substack{S
\subseteq V \\S \text{ with high rank}}} \Gamma^c_S \bigcap \ \bigl(
\bigcup_{\substack{S \subseteq V\\ S \text{ with low rank}}} \Gamma_S \bigr). \nonumber
\end{align}
By
Lemma \ref{gammacloseddeltaopen}, $\Gamma_S$ is a closed set and
$\Gamma^c_S$ is open. Hence, $\Delta_R$ is open. Therefore, there
exists $\delta>0$ such that $\gamma \in \Delta_R$ whenever $\gamma
\geq 0$ and the distance between the two vectors
$d(\tilde{\lambda}, \gamma)< \delta$.

Let $0<\epsilon<\frac{1}{2\sqrt{|V|}} \delta$ and
$\lambda=\tilde{\lambda}+\epsilon e$. Then, the distance between
$\lambda$ and $\tilde{\lambda}$ is $d(\lambda,
\tilde{\lambda})=|\epsilon e|< \frac{1}{2\sqrt{|V|}} \delta |e| =
\frac{1}{2} \delta$. Then, $\lambda \in \Delta_R$ and $\lambda\geq
\epsilon e$.

Now, let $Q(\lambda)=\{S|S \subseteq V, S \text{ with low rank},
\lambda \in \Gamma_S\}$. We will next construct a sequence of
low-rank node sets, $S_i$, for $i = 1, 2, \cdots$. Since each of
them has a low rank, there exists an $|S_i|$-dimensional vector
$g^i \neq 0$ with $||g^i||=1$ such that $(g^i)' e =0$ and $(g^i)'
M_{S_i}=0$. We then extend each $g^i$ to a $|V|$-dimensional
vector by setting the values of the new components to be zero.
With a little abuse of notation, we call this $|V|$-dimensional
vector $g^i$ as well.

We now construct the sequence of $S_i$. If $Q(\lambda) \neq
\emptyset$, pick any subset $S_1 \in Q(\lambda)$. Let
$\lambda^1=\lambda+1/2 \epsilon g^{1}$. Next, if $Q(\lambda^1)
\neq \emptyset$, pick any $S_2 \in Q(\lambda^1)$ and let
$\lambda^2=\lambda^1+1/2^2 \epsilon g^{2}$. In step $j$, if
$Q(\lambda^{j-1}) \neq \emptyset$, we will pick any $S^j \in
Q(\lambda^{j-1})$ and let $\lambda^j= \lambda^{j-1}+1/2^j \epsilon
g^{j}$.  This procedure will go on until $Q(\lambda^j)$ becomes
empty for some $j$. We can check that the $i$th component of
$\lambda^j$ is $\lambda^j(i)=(\lambda+1/2\epsilon g^{1}+1/2^2
\epsilon g^{2}+ \cdots +1/2^j \epsilon g^{j})(i) \geq \epsilon
-1/2 \epsilon-1/2^2 \epsilon- \cdots -1/2^j \epsilon \geq 0$. This
ensures that $\lambda^j$ is always a non-negative vector for all
$j$.

Now, we will show that there exists an integer $K \geq 0$ such
that $Q(\lambda^K)$ becomes empty for the first time (hence, the
sequence of $S_j$ ends at $S_{K-1}$, or contains no elements if
$K=0$). For convenience, let $\lambda^0 = \lambda$.

We will show that $S_j \not \in Q(\lambda^k)$ for $k \geq
j$, where $S_j \neq \emptyset$. Suppose $S_j \in Q(\lambda^k)$ for some $k \geq j$. Then,
$\lambda^k \in \Gamma_{S_j}$, which implies that
$[\lambda^k]_{S_j}=d^1 e+\nu^1$ for some $d^1 \geq 0$ and $\nu^1
\in Co(M_{S_j})$. From the construction procedure, we know that
$S_j \in Q(\lambda^{j-1})$, which implies that
$[\lambda^{j-1}]_{S_j}=d^2 e+\nu^2$ for some $d^2 \geq 0$ and
$\nu^2 \in Co(M_{S_j})$. Since 
\begin{align}
\lambda^k&=\lambda^{j-1}+1/2^j
\epsilon g^{j}+ 1/2^{j+1} \epsilon g^{{j+1}}+ \cdots +1/2^k
\epsilon g^{k}, \nonumber
\end{align} 
we have
\begin{align}
[\lambda^k]_{S_j}&=d^2 e+\nu^2 + 
 [1/2^j \epsilon
g^{j}+ 1/2^{j+1} \epsilon g^{{j+1}}+ \cdots +1/2^k \epsilon
g^{k}]_{S_j}. \nonumber 
\end{align}
Then, 
\begin{align}
(g^{j})' [\lambda^k]_{S_j} 
= & (g^{j})'(d^2 e+\nu^2+ 
 [1/2^j \epsilon g^{j}+ 1/2^{j+1} \epsilon g^{{j+1}}+
\cdots +1/2^k \epsilon g^{k}]_{S_j}) \nonumber \\
= & 1/2^j \epsilon
||g^{j}||^2+1/2^{j+1} \epsilon (g^{j})' [g^{{j+1}}]_{S_j}+ \cdots 
 +1/2^k \epsilon (g^{j})' [g^{k}]_{S_j} \nonumber \\ 
\geq & 1/2^j
\epsilon-1/2^{j+1} \epsilon- \cdots -1/2^k \epsilon > 0. \nonumber
\end{align}
However, since
$[\lambda^k]_{S_j}=d^1 e+\nu^1$, we have $(g^{j})'
[\lambda^k]_{S_j}=(g^{j})'(d^1 e+\nu^1)=0$, leading to a
contradiction. Hence, $S_j \not \in Q(\lambda^k)$ for $k \geq j$.

In summary, each non-empty $S_j$ in the constructed sequence is in
$Q(\lambda^{j-1})$ but not in $Q(\lambda^k)$ for $k \geq j$.
Hence, each $S_j$ is distinct. Since there is a finite number of non-empty
node sets $S \subseteq V$, there exists an integer $K \geq 0$ such
that $Q(\lambda^K)$ becomes empty for the first time.

Then, $\lambda^K \not \in \Gamma_S$ for any node set $S$ with a
low rank. Hence, $\lambda^K \in \bigcap_{S \subseteq V, S \text{
with low rank}} \Gamma^c_S$. The distance between $\lambda$ and
$\lambda^K$ is $d(\lambda, \lambda^K) \leq \epsilon
(1/2+1/2^2+...+1/2^K)< \epsilon$. Then, the distance between
$\tilde{\lambda}$ and $\lambda^K$ is $d(\tilde{\lambda},
\lambda^K)\leq d(\tilde{\lambda}, \lambda)+d(\lambda,
\lambda^K)\leq \epsilon \sqrt{|V|} +\epsilon< \delta$. Hence,
$\lambda^K \in \Delta_R$. It follows $\lambda^K \in \bigcap_{S
\subseteq V}\Gamma^c_S=\Delta_C$.

Since $\epsilon$ can be chosen arbitrarily small,
$\tilde{\lambda}$ is a limit point of $\Delta_C$. Thus,
$\tilde{\lambda} \in \overline{\Delta_C}$, implying
$(\Delta_R-\Delta_C) \subseteq \overline{\Delta_C}$. Hence,
$\overline{(\Delta_R-\Delta_C)} \subseteq \overline{\Delta_C}$.
\end{proof}

The following is an intermediary lemma.
\begin{lemma} \label{lem:lambdagamma}
If a non-empty set $S \subseteq V$ satisfies $\sigma^*_S=1$, then $\Gamma_S \bigcap \Lambda^o=\emptyset$.
\end{lemma}
\begin{proof}
Suppose there exists a vector $\lambda \in \Gamma_S \bigcap
\Lambda^o$. By Definition \ref{def:uniformposidomin},
$[\lambda]_S=\nu+d e$ for some $d \geq 0$ and $\nu \in Co(M_S)$.
Since $\lambda \in \Lambda^o$, by Lemma \ref{lem:subsetcapacity},
$[\lambda]_S+\epsilon e \leq \mu$ for some $\mu \in Co(M_S)$ and a
small enough $\epsilon
> 0$. Hence, $\mu \geq \nu+(d+\epsilon)e \geq \nu(1+d+\epsilon)$.
Thus, $\sigma^*_S<1$ and we arrive at a contradiction.
\end{proof}

\begin{lemma} \label{fullcapacity}
If every high-rank node set $S
\subseteq V$ satisfies $\sigma^*_S=1$, then,
$\overline{\Delta_C}=\overline{\Delta_R}=\Lambda$ and $\Delta_R = \Lambda^o$.
\end{lemma}
\begin{proof}
According to Definition \ref{def:deltasregion}, we have
$\Delta_R=\bigcap_{S \subseteq V, S \text{ with high rank}}
\Gamma_S^c$. For any high-rank node set $S$, since $\sigma^*_S=1$,
we have $\Gamma_S \bigcap \Lambda^o=\emptyset$ by Lemma
\ref{lem:lambdagamma}, which implies $\Gamma^c_S \bigcap
\Lambda^o=\Lambda^o$. Hence, $\Delta_R \bigcap
\Lambda^o=\bigcap_{S \subseteq V, S \text{ with high rank}}
\Gamma_S^c \bigcap \Lambda^o=\Lambda^o$. Combining this with
Theorem \ref{thm:deltacr}, we get
$\overline{\Delta_C}=\overline{\Delta_R}=\overline{\Lambda^o}=\Lambda$. Also, the fact that $\Delta_R \bigcap \Lambda^o = \Lambda^o$ implies $\Lambda^o \subseteq \Delta_R$. Since $\Delta_R$ is an open set in $\Lambda$ and $\Lambda^o$ is the largest open set in $\Lambda$, it must be that $\Delta_R = \Lambda^o$.
\end{proof}

\noindent  {\bf Remark.} From Lemma \ref{fullcapacity}, we know that when all the subsets $S
\subseteq V$ satisfy either set local pooling (i.e., $\sigma^*_S=1$) or the rank of $S$ is
low, then $\Delta_R = \Lambda^o$. That is, the entire $\Lambda^o$ is achievable by LQF, assuming the arrival processes have non-zero variances. This is the same statement as Theorem 1 of \cite{DW06}. Thus, the newly developed theory here is able to reproduce the main result of \cite{DW06}.

\begin{lemma} \label{lem:dconvexcond}
$\overline{\Delta_C}=\Lambda$ if and only if $\overline{\Delta_C}$ is convex. Similarly, $\overline{\Delta_R}=\Lambda$ if and only if $\overline{\Delta_R}$ is convex.
\end{lemma}
\begin{proof}
It is obvious that $\overline{\Delta_C}=\Lambda$ implies $\overline{\Delta_C}$ is convex. We will next show the converse.
Since $\Delta_C \subseteq \Lambda$ and $\Lambda$ is a closed set, we have $\overline{\Delta_C} \subseteq \Lambda$. Because $\Omega^o \subseteq \Delta_C$, we have $\Omega \subseteq \overline{\Delta_C}$. Since $\Omega$ contains all the extreme points of $\Lambda$ (Corollary \ref{cor:extpt}), $\overline{\Delta_C}$ also contains all of them. Since $\Lambda$ is the convex combination of all its extreme points and $\overline{\Delta_C}$ is convex, we must have $\Lambda \subseteq \overline{\Delta_C}$.

The second statement can be proved similarly.
\end{proof}


%

%% file: sections/hypergraph.tex
\label{sec:hypergraph}

The scheduling problem in this paper is deeply connected with graph
coloring and its related problems. In this section, we
will introduce fractional coloring, and more generally, aspects of the fractional graph theory that can provide useful tools for studying the stability regions discussed in the previous sections.

\subsection{Fractional Coloring and Capacity Region}
\label{sec:fcandcap}

The chromatic number of a graph $G$, denoted by $\chi(G)$, is the
minimum number of colors needed to paint the nodes so that the connected
nodes do not share the same color. When we relax the integrality
constraints of the chromatic number problem and introduce a parameter $\lambda \in \mathbb{R}_+^{|V|}$, we obtain the
following linear programming (LP) problem.

\begin{definition}
Given a graph $G=(V,E)$ and $\lambda \in \mathbb{R}_+^{|V|}$, the \emph{weighted
fractional coloring} problem with the weight vector $\lambda$ is:
\begin{align}
      \chi_f (G,\lambda) \triangleq \min \ e' \alpha,
      \text{ subject to } M_V \alpha \geq \lambda,
                        \alpha \geq 0. \label{eq:wfc}
\end{align}
\end{definition}

The optimal value of the above problem, $\chi_f (G,\lambda)$, is called the weighted
fractional chromatic number, which is known to be related to
the capacity region as follows (see \cite{Cem10}):
\begin{align}
\Lambda=\{\lambda \in \mathbb{R}_+^{|V|} \ | \ \chi_f(G, \lambda) \leq 1\}. \label{eq:capwfc}
\end{align}
Based on (\ref{eq:wfc}), $\chi_f (G,\lambda)$ can be interpreted as the fastest way of serving
queued data when the queue sizes are proportional to the weights $\lambda$. Based on (\ref{eq:capwfc}), $\chi_f (G,\lambda)$ can be interpreted as the `traffic load' to the network.

The relevance and usefulness of this problem to the study of wireless scheduling have been amply demonstrated in \cite{Cem10}. The characterization of the capacity region by (\ref{eq:capwfc}) suggests that the fractional chromatic
number can serve as an oracle for judging whether an arrival rate
vector is in the capacity region or not. With this observation and with known complexity results about the fractional coloring problem, the authors of \cite{Cem10} have derived results about the inherent complexity of the wireless scheduling problem.

\subsection{Weighted Fractional Matching Number and $\Omega$ Region}

%
We next discuss the problem of finding the weighted fractional
matching number of a graph \cite{SU97}. This problem can help to
decide whether a vector is in $\Omega$. 


\begin{definition} \label{def:wfm}
Given a graph $G=(V,E)$ and $\lambda \in \mathbb{R}_+^{|V|}$, the \emph{weighted
fractional matching number} problem with the weight vector
$\lambda$ is:
\begin{align}
      \phi_f (G,\lambda) \triangleq \max \ e' \beta,
      \text{ subject to } M_V \beta \leq \lambda,
                        \beta \geq 0. \label{eq:wfm}
\end{align}
\end{definition}

The above problem is the Lagrangian dual of the weighted
fractional transversal number problem, which is the hypergraph
dual problem of the weighted fractional coloring problem
\cite{SU97}. Here, the $i$th component of $\beta$ can be interpreted as
the amount of time for which the $i$th maximal schedule is used.
The weighted fractional
matching number, $\phi_f (G,\lambda)$, can be interpreted as the slowest way of serving
the queued data (in the amount $\lambda$) using only the maximal
schedules, subject to the additional constraint that a schedule
should not be selected if it activates a link associated with an empty queue.
%

\begin{lemma} \label{lem:omegatransversal}
The $\Omega$ region satisfies the following:
\begin{align*}
 \Omega=\{\lambda \in \mathbb{R}_+^{|V|} \ | \ \phi_f(G_S, [\lambda]_S) \leq 1, \forall S \subseteq V, S \neq \emptyset\}.
\end{align*}
\end{lemma}
\begin{proof}
Consider any vector $\lambda \in \Omega$ and an arbitrary
non-empty node set $S \subseteq V$. Suppose $\phi_f(G_S,
[\lambda]_S)$ $>1$. Then, by Definition \ref{def:wfm}, we have $[\lambda]_S \geq k \nu$ for some $\nu \in
Co(M_S)$ and $k>1$. Let $Z$ be the largest subset in $S$ such that
$[\nu]_Z=0$. Note that $Z \neq S$. Then, the vector $[\nu]_{S-Z} \in Co(M_{S-Z})$ and
$[\nu]_{S-Z}>0$. Hence, $[\lambda]_{S-Z}\geq k [\nu]_{S-Z}
>[\nu]_{S-Z}$, which implies that $\lambda \in \Pi_{S-Z}$. According
to Definition \ref{def:posidomin} and \ref{def:omegaregion},
$\lambda \not \in \Omega$.

Conversely, suppose a vector $\lambda$ satisfies $\phi_f(G_S,
[\lambda]_S) \leq 1$ for every non-empty $S \subseteq V$. Then,
$[\lambda]_S \not
> \nu$ for any $\nu \in Co(M_S)$. Otherwise, there would exist $k>1$
such that $[\lambda]_S \geq k \nu$ for some $\nu \in Co(M_S)$, which implies that
$\phi_f(G_S, [\lambda]_S) >1$. Thus, $\lambda \in \Omega$.
\end{proof}

\subsection{Hypergraph Duality and Set $\sigma$-local Pooling}

We next relate the ratio of $\chi_f (G,\lambda)$ and $\phi_f (G,\lambda)$ to the set $\sigma$-local pooling factor. First, we have the following lemma.


\begin{lemma} \label{lem:philinear}
For any $k \geq 0$, $\chi_f(G, k \lambda)=k \chi_f(G, \lambda)$ and $\phi_f(G, k
\lambda)=k \phi_f(G, \lambda)$.
\end{lemma}
\begin{proof}
The case of $k = 0$ is trivial. We only focus on the case of $k
> 0$. Suppose $\beta^*$ is an optimal solution to the problem in (\ref{eq:wfm}) for finding $\phi_f(G, \lambda)$. Then, $k \beta^*$ is feasible to the problem for finding $\phi_f(G, k
\lambda)$. Since $e' (k \beta^*)= k e'
\beta^*$, $\phi_f(G, k \lambda) \geq k \phi_f(G, \lambda)$.

Conversely, suppose $\tilde{\beta}$ is an optimal solution to the
problem for finding $\phi_f(G, k \lambda)$. Then, $\tilde{\beta}/k$ is feasible to the problem for finding $\phi_f(G, \lambda)$. Since
$e' (\tilde{\beta}/k)= e' \tilde{\beta}/k$, $\phi_f(G, \lambda) \geq 
\phi_f(G, k \lambda)/k$. Therefore, $\phi_f(G, k \lambda)=k
\phi_f(G, \lambda)$, for $k>0$. A similar argument can be used to show $\chi_f(G, k \lambda)=k \chi_f(G, \lambda)$.
\end{proof}




\begin{theorem} \label{thm:setsigmaratio}
Given a non-empty node set $S \subseteq V$, the set $\sigma$-local
pooling factor of $S$ satisfies the following\footnote{We take the convention
$a/0 = \infty$ for any scalar $a \geq 0$. Note that, if any component of $\lambda$ is equal to zero, then $\phi_f(G_S, \lambda) = 0$. As a result, the optimal solution $\lambda^*$ to (\ref{eq:sigsratio}) must satisfy $\lambda^* > 0$.}:
\begin{align}
\sigma^*_S=\min_{\lambda \geq 0} \frac{\chi_f(G_S, \lambda)}{\phi_f(G_S,
\lambda)}. \label{eq:sigsratio}
\end{align}
\end{theorem}
\begin{proof}
Suppose $(\sigma^*_S, \mu^*_S, \nu^*_S)$ is an optimal solution to
the problem in (\ref{eq:setsigmapro}). Then,
we choose $\lambda=\nu^*_S$. Since $\sigma^*_S \mu^*_S \geq
\nu^*_S$, we have 
\begin{align}
\chi_f(G_S, \lambda) & = \chi_f(G_S, \nu^*_S) \nonumber \\
& \leq \chi_f(G_S, \sigma^*_S \mu^*_S) 
=\sigma^*_S \chi_f(G_S, \mu^*_S) \leq \sigma^*_S. \nonumber
\end{align}
The last inequality above uses the fact $\chi_f(G_S, \mu^*_S) \leq 1$, which follows from (\ref{eq:capwfc}) (since $\mu^*_S \in \Lambda_S$).

Because $\nu^*_S \in Co(M_S)$, there exists a
non-negative vector $\beta$ such that $M_S \beta= \nu^*_S$ and
$e' \beta=1$. Such $\beta$ is feasible to (\ref{eq:wfm}) for finding $\phi_f(G_S, \lambda)$. Hence, $\phi_f(G_S, \lambda) \geq 1$.
Therefore, $\min_{\lambda \geq 0} \chi_f(G_S, \lambda)/\phi_f(G_S,
\lambda) \leq \sigma^*_S$.

Next, suppose $\lambda^*$ is an optimal solution for the problem
$\min_{\lambda \geq 0} \chi_f(G_S, \lambda)/\phi_f(G_S, \lambda)$.
Suppose $\alpha^*$ and $\beta^*$ are optimal solutions for the
problems of finding $\chi_f(G_S, \lambda^*)$ and $\phi_f(G_S,
\lambda^*)$, respectively. Then, we have $\min_{\lambda \geq 0}
\chi_f(G_S, \lambda)/\phi_f(G_S, \lambda)=\sum_i \alpha^*_i/\sum_i
\beta^*_i$. Now, let $\mu= M_S \alpha^*/\sum_i \alpha^*_i$ and
$\nu= M_S \beta^*/\sum_i \beta^*_i$. Then, $\mu, \nu \in Co(M_S)$
and 
\begin{align}
\frac{\chi_f(G_S, \lambda^*)}{\phi_f(G_S,
\lambda^*)} \mu=\frac{\sum_i \alpha^*_i}{\sum_i \beta^*_i} \frac{M_S
\alpha^*}{\sum_i \alpha^*_i}= \frac{M_S \alpha^*}{\sum_i \beta^*_i}. \nonumber
\end{align}
By
the feasibility of $\alpha^*$ and $\beta^*$ to (\ref{eq:wfc}) and (\ref{eq:wfm}), respectively, $M_S \alpha^* \geq
\lambda^* \geq M_S \beta^*$. Hence, 
\begin{align}
\frac{M_S \alpha^*}{\sum_i \beta^*_i} \geq \frac{\lambda^*}{\sum_i \beta^*_i} \geq \frac{M_S \beta^*}{\sum_i \beta^*_i}= \nu. \nonumber
\end{align}
Thus, $\chi_f(G_S, \lambda^*)/\phi_f(G_S, \lambda^*)$ is feasible to the problem in
(\ref{eq:setsigmapro}). Therefore,
$\min_{\lambda \geq 0} \chi_f(G_S, \lambda)/\phi_f(G_S, \lambda)
\geq \sigma^*_S$.
\end{proof}

The theorem above shows that the set $\sigma$-local pooling factor is the same as the minimum of the hypergraph duality ratios over all different weights.

\subsection{Weighted Fractional Domination Number and $\Delta_C$ Region}

\begin{definition} \label{def:weightedfrac}
Given a graph $G=(V, E)$ and $\lambda \in \mathbb{R}_+^{|V|}$, the
\emph{weighted fractional domination number} problem with
the weight vector $\lambda$ is:
\begin{align}
      \tau_f (G,\lambda) \triangleq \max d,
      \text{ subject to } de+\nu=\lambda, 
                          \nu \in Co(M_V). \nonumber
\end{align}
\end{definition}
For convenience, let $\tau_f(G, \lambda)=-\infty$ when the problem
is infeasible.


By Definition \ref{def:uniformposidomin}, \ref{def:deltacregion}
and \ref{def:weightedfrac}, we have the following lemma.
\begin{lemma} \label{lem:deltatau}
The following relations hold: 
\begin{align*}
 &\Gamma_S=\{\lambda \in \mathbb{R}_+^{|V|} \ | \ \tau_f(G_S, [\lambda]_S) \geq
 0, \text{  for } S \subseteq V, S \neq \emptyset\}, \\
 &\Delta_C=\{\lambda \in \mathbb{R}_+^{|V|} \ | \ \tau_f(G_S, [\lambda]_S) < 0, \forall S \subseteq V,
 S \neq \emptyset\}.
\end{align*}
\end{lemma}

%% file: sections/simulations.tex
\label{sec:experiment}

In this section, we show some simulation results. The main purpose is to confirm some of the less intuitive theoretical results. We first show the performance of LQF on the six-cycle graph, denoted
by $C_6$, for arrival rate vectors in different sections of the capacity
region. For $C_6$, LQF can achieve the entire interior
of the capacity region for arrivals satisfying assumption A1 and with
non-zero variances. On the other hand, for constant arrivals,
experiments have shown that some rate vectors in the interior of
the capacity region are not achievable by LQF. 


In the experiments with constant arrivals, we use a \emph{load} parameter to scale the
arrival rate vectors. Each experiment runs for $10^6$ iterations
with initial queue sizes of $10^3$. The following arrival rate vectors
are used for the results in Fig. \ref{fig:LQF.unstable.C_6}:
\begin{align}
\lambda^1 & = (\frac{1}{2}-\epsilon,\frac{1}{2}-\epsilon,\frac{1}{2}-\epsilon,\frac{1}{2}-\epsilon, \frac{1}{2}-\epsilon,\frac{1}{2}-\epsilon)' \nonumber \\ 
\lambda^2 & = (\frac{1}{2}-\epsilon,\frac{1}{2}-2\epsilon,\frac{1}{2}-3\epsilon,
\frac{1}{2}-\epsilon,\frac{1}{2}-2\epsilon,\frac{1}{2}-3\epsilon)' \nonumber \\
\lambda^3 & = (\frac{1}{2}-\epsilon,\frac{1}{2}-2\epsilon,
\frac{1}{2}-2\epsilon,\frac{1}{2}-2\epsilon,\frac{1}{2}-2\epsilon,\frac{1}{2}-2\epsilon)' \nonumber,
\end{align}
where $\epsilon = 10^{-3}$. Note that $0.7\lambda^1 < 0.95\lambda^2
< \lambda^3$. However, judging by the queue sizes in Fig. \ref{fig:LQF.unstable.C_6}, the arrival rate vectors $0.7\lambda^1$ and $0.95\lambda^2$ seem to be
not stabilizable, whereas $\lambda^3$ seems to be stabilizable. The
theory allows this counter-intuitive phenomenon. Readers can verify
that $0.7\lambda^1, 0.95\lambda^2 \notin \Delta_C$ while $\lambda^3
\in \Delta_C$.


\begin{figure}[htpb]
\centering
\includegraphics[width=3.5in]{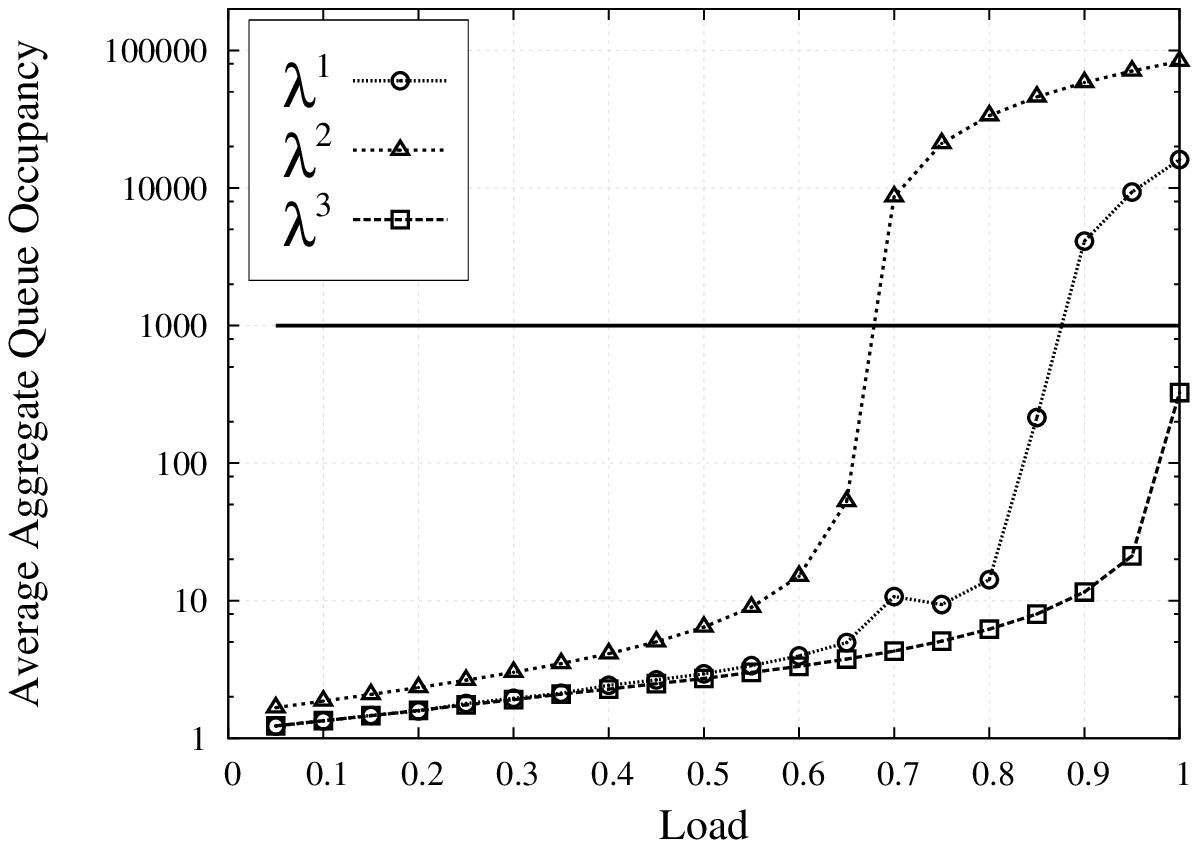}
\caption{Constant arrivals in $C_{6}$}
\label{fig:LQF.unstable.C_6}
\end{figure}

In Section \ref{sec:delta}, we generalize the definitions of high or low-rank
graphs. In Fig. \ref{fig:eight.and.six.cycles}, we provide an
interference graph that is not set local pooling ($\sigma_S^* < 1$)
and is of high rank according to the original definition in \cite{DW06}.
However, in the new definition, the graph is of low rank. We ran
simulations with initial queue sizes of $10^3$ using the Bernoulli
arrivals with an identical arrival rate of 0.499. Fig.
\ref{fig:eight.and.six.plot} shows the evolution of the average queue size 
for nodes 1-8 and 9-13 over $10^6$ iterations. The queues for nodes
1-8 appear to be unstable and the queues for nodes 9-13 appear to be
stable. Our refinement of the rank condition rules out the
possibility that the queues of all nodes are simultaneously the
longest and remain longest, whereas the previous rank condition does not
rule that out.

\begin{figure}[htpb]
\centering
\includegraphics[width= 2.5in]{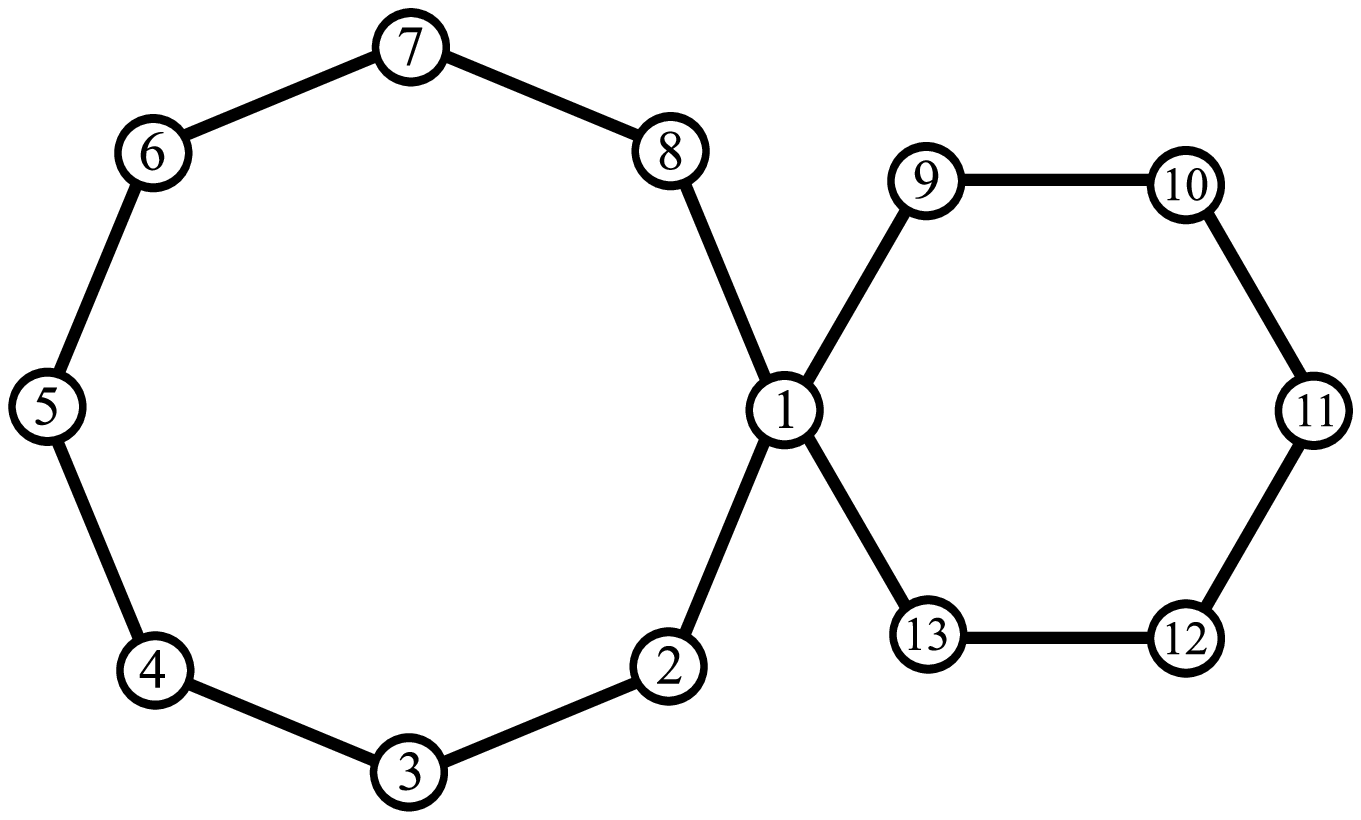}
\caption{An interference graph with $C_6$ connected to $C_8$. For
this graph, the ranks are $R(M_V) = R(M_V,e) = 12$.}
\label{fig:eight.and.six.cycles}
\end{figure}

\begin{figure}[htpb]
\centering
\includegraphics[width= 3.5in]{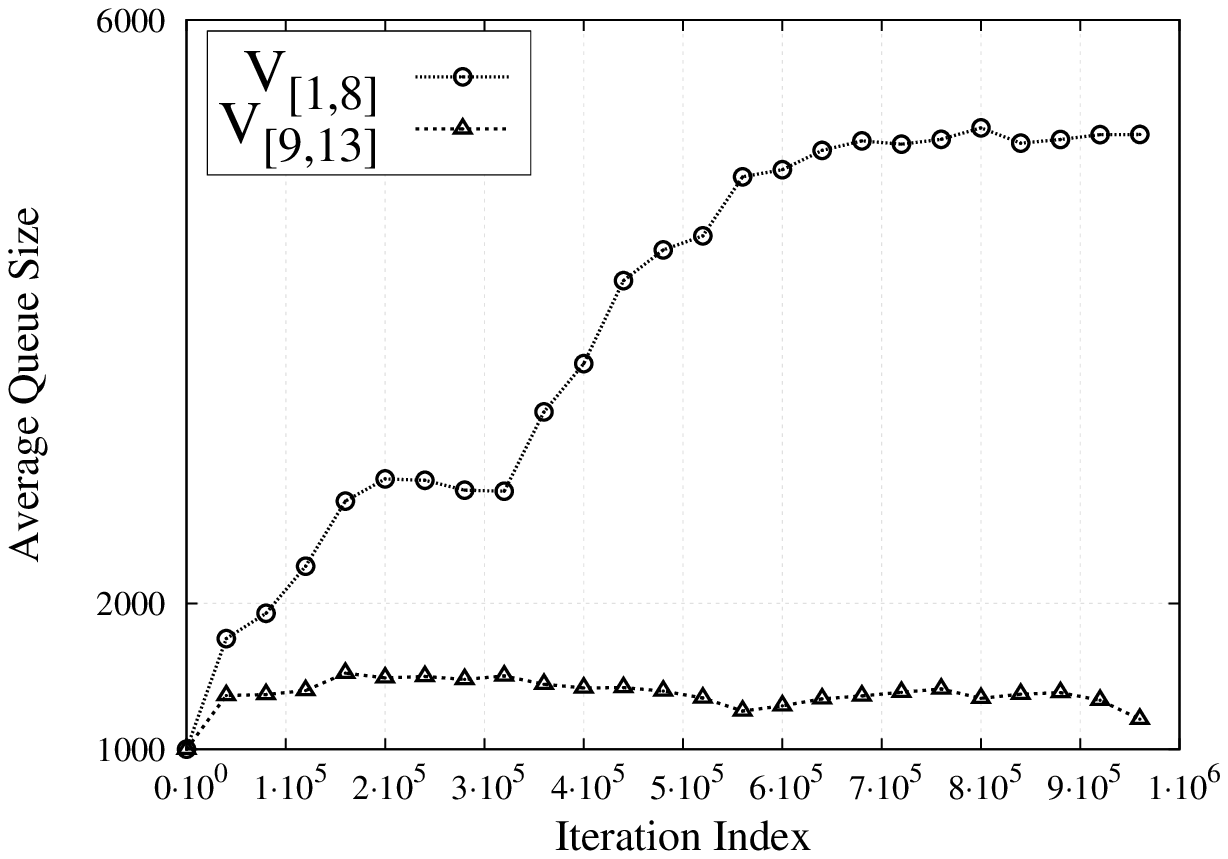}
\caption{Average queue sizes for nodes 1-8 on $C_8$ (labelled as
$V_{[1,8]}$) and for the nodes 9-13 on $C_6$ ($V_{[9,12]}$).}
\label{fig:eight.and.six.plot}
\end{figure}

%% file: sections/conclusion.tex
\label{sec:conclusion}

In this paper, we investigate the performance guarantee of the LQF
scheduling policy in wireless networks. The objective is to discover new stability regions of LQF that are larger than those previously known, and to improve our knowledge about the largest possible stability region of LQF. We show that it is necessary to go beyond the existing framework of linear reduction of the capacity region, and move to a non-linear framework.

We introduce the concepts
of strictly dominating vectors and uniformly dominating
vectors; the former leads to the new stability region of LQF, $\Omega$, and the latter leads the stability regions $\Delta_C$ and $\Delta_R$. We show that $\Omega$ contains $\Sigma^*(G) \Lambda$, which is the stability region given in
\cite{BCY09}. We also show $\Omega^o \subseteq \Delta_C \subseteq \Delta_R$. Hence, the new stability regions all capture the performance of LQF better. Contrary to the previously-known regions of stability, the closures of these new stability regions contain all the extreme points of the capacity region $\Lambda$, but they are not convex in general. The only case where they are convex is when they are equal to the capacity region itself, which occurs only for selected interference graphs. The general lack of convexity is not a defect of the theory. We show that, when LQF cannot achieve the full capacity region, the largest achievable region cannot be convex.

The study reveals a counter-intuitive situation where increasing the arrival rates helps LQF to stabilize the network. It turns out, in this case, the
original rate vector is outside $\Delta_C$, and after the
rate increase, the new rate vector is inside $\Delta_C$.
We also generalize the rank condition studied in \cite{DW06}, and
with this generalization, refine the stability results for
non-deterministic arrivals. We can show that if a set of nodes
satisfies the new low-rank condition, the queue sizes of these
nodes will be separated. Based on this result, we can enlarge
$\Delta_C$ to $\Delta_R$, which is achievable by LQF under
non-deterministic arrivals. Interestingly, we show that the
closures of $\Delta_C$ and $\Delta_R$ are the same. Finally, we
introduce several linear programming problems encountered in the
fractional graph theory, which can provide tools for studying the
newly developed stability regions. We show that a ratio between
the weighted fractional coloring number and the weighted
fractional matching number is related to the set $\sigma$-local
pooling factor introduced in \cite{BCY09}.


%% file: sections/appendix.tex
\label{sec:appendix}


\begin{proof}[Proof of Lemma \ref{indepextrem}]
Suppose $\nu$ is an independent set of $G$, represented by a 0-1
vector. Clearly, $\nu \in
\Lambda$. Let us write $\nu=a \nu^1+ (1-a) \nu^2$ for some
$0 \leq a \leq 1$, and $\nu^1, \nu^2 \in \Lambda$. Note that $\mu \leq e$
for every $\mu \in Co(M_V)$. Thus, $0\leq \nu^1 \leq e$ and $0\leq
\nu^2 \leq e$. For any index $i$, if $\nu_i=0$, we must have
$\nu^1_i=0$ and $\nu^2_i=0$. Similarly, if $\nu_i=1$,
we have $\nu^1_i=\nu^2_i=1$. Therefore, $\nu=\nu^1=\nu^2$ and $\nu$ is an extreme point of $\Lambda$.

Conversely, take any extreme point $\nu$ of $\Lambda$. Then, $\nu
\leq \mu$ for some $\mu \in Co(M_V)$. For an index $i$, if $\nu_i
>0$, we claim that $\nu_i = \mu_i$. Otherwise, we let $t=\mu_i-\nu_i>0$.
Then, we can create $\bar{\nu}$ and $\tilde{\nu}$ such that
$\bar{\nu}_j=\tilde{\nu}_j=\nu_j$ for $j \neq i$. We can find an
$\epsilon>0$ such that $\bar{\nu}_i \triangleq \nu_i-\epsilon t
\geq 0$, and we let $\tilde{\nu}_i=\mu_i > 0$. Since $0 \leq
\bar{\nu} \leq \tilde{\nu} \leq \mu$, we have $\bar{\nu}, \tilde{\nu} \in
\Lambda$. It is easy to see $\nu=\frac {1} {1+\epsilon} \bar{\nu}+
\frac{\epsilon} {1+\epsilon} \tilde{\nu}$, which implies $\nu$ is
not an extreme point.  Hence, either $\nu_i=0$ or $\nu_i=\mu_i$,
for all $i$.

We now show that $\nu_i=0$ or $\nu_i=\mu_i=1$ for all $i$. Write
$\mu$ as $\mu = \sum_{j=1}^k a_j \mu^j$, where each $\mu^j \in M_V$, each $a_j>0$, and $\sum_{j=1}^k a_j =1$. Let $\nu^j_i=0$ if $\nu_i=0$;
$\nu^j_i=\mu^j_i$ otherwise, for all $0 \leq j \leq k$. Because
$\mu^j \in M_V$, we have $\nu^j_i=0$ or $\nu^j_i=1$. It is easy to
check that $\nu= \sum_{j=1}^k a_j \nu^j$. Since $\nu^j \leq \mu^j$, we have $\nu^j \in \Lambda$ for each $j$.
Since
$\nu$ is an extreme point of $\Lambda$, $\nu^1=\nu^2=\cdots =\nu^k$.
Thus, $\nu_i=1$ or $\nu_i=0$ for all $i$.

It is easy to see that any 0-1 vector in $Co(M_V)$ must be a feasible schedule, i.e., an independent set of $G$.
Since $\mu \in Co(M_V)$, the set of nodes, $S$, for
which $\mu_i = 1$ forms an independent set. Let $S'$ be the set of
nodes for which $\nu_i = 1$. We have $S' \subseteq S$. Therefore,
$\nu$ corresponds to an independent set.
\end{proof}